\documentclass[12pt,a4paper]{article}

\usepackage{amsmath,amssymb,amsthm,mathtools}
\usepackage{mathrsfs}
\usepackage{bm}
\usepackage{cite}
\usepackage{tikz}
\usetikzlibrary{arrows.meta,calc,decorations.pathmorphing}
\usepackage{geometry}
\usepackage[colorlinks=true,linkcolor=blue,citecolor=blue,urlcolor=blue]{hyperref}
\newcommand{\Tr}{\operatorname{Tr}}

\newcommand{\Ad}{\operatorname{Ad}}
\newcommand{\Id}{\operatorname{Id}}
\newcommand{\Iplus}{\mathscr I^+}
\newcommand{\Hplus}{\mathcal H^+}
\newcommand{\A}{\mathcal A}

\newcommand{\E}{\mathcal E}
\newcommand{\I}{\mathcal I}
\newcommand{\N}{\mathcal N}

\newcommand{\F}{\mathcal F}
\newcommand{\K}{\mathsf K}
\newcommand{\Hil}{\mathcal H}

\newcommand{\hard}{\mathrm{hard}}
\newcommand{\tot}{\mathrm{tot}}
\newcommand{\ch}{\mathrm{ch}}
\newcommand{\matter}{\mathrm{matter}}
\newcommand{\Coul}{\mathrm{Coul}}
\newcommand{\open}{\mathrm{open}}
\newcommand{\close}{\mathrm{close}}
\newcommand{\opt}{\mathrm{opt}}
\newcommand{\admissible}{\mathrm{admissible}}
\newcommand{\out}{\mathrm{out}}
\newcommand{\late}{\mathrm{late}}
\newcommand{\lab}{\mathrm{lab}}
\newcommand{\env}{\mathrm{env}}
\newcommand{\inm}{\mathrm{in}}
\newcommand{\rad}{\mathrm{rad}}

\newcommand{\hor}{\mathrm{H}}
\newcommand{\unres}{\mathrm{unres}}
\newcommand{\dress}{\mathrm{dress}}

\newcommand{\dd}{\mathrm d}
\newcommand{\ii}{\mathrm i}

\newcommand{\half}{\frac12}

\theoremstyle{plain}
\newtheorem{proposition}{Proposition}
\newtheorem{theorem}{Theorem}
\theoremstyle{remark}
\newtheorem{remark}{Remark}
\newtheorem{definition}{Definition}
\newtheorem{problem}{Problem}

\title{\textbf{Horizon-Restricted Leading Soft QED \\
as \\
Open Quantum System}}
\author{\bf Soo-Jong Rey\\[0.2cm]
{\sl Kwangwoon University}\\
{\sl Seoul, Korea}}
\date{\tt sjrey@kw.ac.kr}

\begin{document}
\maketitle

\begin{abstract}
I formulate black-hole-horizon-induced decoherence of charged branch codes as the leading-soft QED restricted to an exterior algebra, formulated as an open quantum system. The fixed-history Feynman--Vernon identity $\F[J,J]=1$ remains exact.  Decoherence enters through the unequal-history influence factor that survives exterior monitoring and belongs to the complementary horizon output.  In the coherent eikonal regime, I derive the completely positive Schur channel $({\cal E}_H^{(0)}\rho)_{ab}=\langle\Phi_b^{H,(0)}|\Phi_a^{H,(0)}\rangle \, \rho_{ab}$. 
The leading soft input is the eikonal factor, projected onto the horizon radiative algebra.  The channel yields Gram-positivity constraints, an exterior quantum-eraser bound, finite-time non-Markovianity tests, soft/hard scaling criteria, and a charged-qutrit interferometer measuring a leading-soft Bargmann holonomy.  The holonomy phase is the rephasing-invariant symplectic area of a triangle in horizon soft phase space. I show that its orientation, common-mode, triangulation, and completely positive determinant identities render falsifiable tests beyond pairwise two-path visibility.
\end{abstract}

\section{Introduction}

Long-range gauge fields necessitate a rigorous separation between the constrained Coulomb field and independent radiation~\cite{BlochNordsieck,YFS,Kinoshita,LeeNauenberg,Chung,KulishFaddeev}. In this work, I employ this separation to formulate black-hole-horizon-induced decoherence as a leading-soft~\cite{WeinbergSoft} open-system channel~\cite{Rey2026a,Rey2026b}. A dressed charged branch carries the Coulomb field required by Gauss's law; the environment consists exclusively of soft radiative records that lie outside the algebra controlled by the exterior experimenter. In flat spacetime, an adiabatic protocol drives the independent radiative record to vanish upon recombination. Near a horizon, however, this identical branch-dependent leading-soft record bifurcates: exterior modes propagate to \(\Iplus\), whereas a complementary component crosses \(\Hplus\) and remains inaccessible to exterior operations.

Existing analyses of horizon decoherence have established the physical effect for charged and massive bodies near black holes and Killing horizons, connected it to local low-frequency two-point functions, and identified the finite-time recovery problem~\cite{UnruhFalse,DSW2022,DSW2023,DSWLocal2025,DKSW2025,DSSoftInfo2025}. (For recent follow-up works, see ~\cite{GrallaWei,LiRN,WilsonGerow,BiggsMaldacena}). I reformulate these elements within the framework of leading-soft inclusive QED and open quantum systems, extending the work of~\cite{Rey2026a,Rey2026b}. The central object of study is a reduced channel specified by an accessible exterior algebra, detector resolution, and a complementary output; it is not merely a coordinate-frequency designation for a horizon detector.

I initiate the construction by assuming a tensor product between the dressed charged branch sector and the independent radiation field,
\begin{equation}
       \rho_{\rm tot}^{\rm in}
       =
       \rho_{\rm br}^{\rm in}\otimes\rho_{\rm rad}^{\rm in},
       \label{eq:intro_product_revised}
\end{equation}
while the branch density matrix itself remains entirely arbitrary. A diagonal branch input generates classical emission histories, whereas a coherent branch input probes the off-diagonal entries of the identical channel. Consequently, the familiar two-branch spatial superposition functions merely as a basis choice for extracting a single channel coefficient, rather than as a fundamental structural assumption. The normalized Feynman--Vernon influence functional obeys the fixed-history identity~\cite{Rey2026a,Rey2026b}
\begin{equation}
       \F[J,J]=1,
       \label{eq:intro_fixed_history_revised}
\end{equation}
and the black-hole effect manifests as the residual unequal-history factor $\F[J_a,J_b]$, which exterior operations are unable to erase.

Within the coherent-state regime, the horizon component of this residual factor assumes the form of a Gram matrix:
\begin{equation}
       G^{H,(0)}_{ab}
       =
       \langle\Phi_b^{H,(0)}|\Phi_a^{H,(0)}\rangle,
       \qquad
       (\E_H^{(0)}\rho)_{ab}=G^{H,(0)}_{ab}\rho_{ab} .
       \label{eq:intro_schur_revised}
\end{equation}
Thus, the horizon contribution acts as a completely positive Schur multiplier on any finite branch code. This channel formulation enforces Gram positivity, constrains exterior quantum erasure, and unifies the horizon coherent-state description with the local two-point-function representation. The identical kernel manifests equivalently as a horizon overlap, a one-particle displacement norm, or a noise kernel smeared with current differences.

This quantum channel perspective yields concrete operational predictions. Exterior monitoring eliminates accessible infrared records yet preserves a horizon-fidelity ceiling. Finite-time protocols exhibit non-completely-positive-divisible recovery when subsequent current motion cancels earlier soft displacements. Soft horizon records and ordinary hard absorption are distinguished by their dwell-time and switching-time scaling. Material systems mimic the electromagnetic channel precisely when their low-frequency kernels coincide when evaluated on the branch-current differences under consideration. The most discriminating observable emerging from this leading-soft construction possesses a multi-branch structure: a three-arm charged interferometer reconstructs the Bargmann product~\cite{Bargmann,Uhlmann}
\begin{equation}
       \mathcal B_{123}=G_{12}G_{23}G_{31},
       \label{eq:intro_bargmann_revised}
\end{equation}
from qutrit tomography as $\rho_{12}\rho_{23}\rho_{31}/(\rho_{11}\rho_{22}\rho_{33})$.   Following ideal exterior erasure, this cyclic visibility isolates the leading horizon Bargmann invariant within the factorized regime. Its phase corresponds to a symplectic area in the horizon soft phase space and provides orientation, common-mode, triangulation, and complete-positivity determinant tests, which remain inaccessible in a simpler pairwise visibility analysis.

The remainder of this paper is organized as follows. I first establish the horizon-decoherence constraints for a black hole and construct the dressed branch space, the exterior algebra, the initial product-state channel, and the Feynman--Vernon kernel. I then connect this construction to the flat-space leading-soft QED open quantum system channel and project the eikonal soft factor onto the horizon algebra. Subsequent sections discuss the local-noise representation, exterior instruments, recovery bounds, finite-time control, soft/hard scaling, and material mimicry as consequences of a unified reduced channel. Finally, I develop falsifiable predictions derived from the charged-qutrit Bargmann interferometer. Technical details are collected in the appendices. Appendix~\ref{app:asrd} formulates the Hayden--Harlow exact-completion issue rigorously as an asymptotic soft-record decoding problem, while Appendix~\ref{app:scope_checks} details consistency checks. Additional appendices cover the coherent-state overlap, the Schur-channel Kraus representation, and the relation to infrared-finite QED.

\section{Diagnostics of horizon-decoherence constraints}

I first establish the constraints implemented by the leading-soft channel construction. In flat spacetime, the late-time electromagnetic field separates into a constrained Coulomb component, which follows the final charge and belongs to the dressed charged state, and an independent radiative component at $\Iplus$.  A two-branch probe state assumes the asymptotic form:
\begin{equation}
       |\Psi_{\tot}\rangle
       =
       {\frac{1}{\sqrt2}}
       \left(
          |1\rangle_{\ch}\otimes|\Psi_1\rangle_{\Iplus}
          +
          |2\rangle_{\ch}\otimes|\Psi_2\rangle_{\Iplus}
       \right),
\end{equation}
so the off-diagonal charged coherence is modulated by the overlap $\langle\Psi_2|\Psi_1\rangle_{\Iplus}$.   Adiabatic switching on and off drives this overlap toward unity in flat spacetime. This recovery baseline eliminates spurious Coulomb distinguishability from the analysis and identifies the independent radiation as the relevant environment.

A stationary black hole introduces a second asymptotic radiative output at $\Hplus$.  For a given branch $a$, the late state assumes the schematic form
\begin{equation}
       |a\rangle_{\ch}\otimes
       |\Psi_a\rangle_{\Iplus}\otimes
       |\Phi_a\rangle_{\Hplus} .
\end{equation}
The exterior adiabatic protocol suppresses radiation to infinity while preserving a nontrivial horizon overlap.  Within the coherent-state approximation,
\begin{equation}
       \left|\langle\Phi_b|\Phi_a\rangle_{\Hplus}\right|
       =
       \exp\left[-\half\|\alpha_a^{\hor}-\alpha_b^{\hor}\|^2\right],
       \label{eq:coherent_overlap_abs}
\end{equation}
where $\alpha_a^{\hor}$ denotes the horizon one-particle wavefunction generated by the branch-dependent retarded solution. The squared norm corresponds to the expected number of entangling horizon photons. The electromagnetic Schwarzschild estimate for a representative two-branch protocol yields
\begin{equation}
       \langle N\rangle_{\hor}
       \sim
       \frac{G^3M^3 q^2 d^2}{\hbar c^6\epsilon_0 b^6}\,T,
       \label{eq:schwarzschild_scaling_restored}
\end{equation}
up to order-unity constants and branch-protocol normalization factors.  Here, $M$ denotes the black-hole mass, $q$ the charge, $d$ the branch separation, $b$ the characteristic distance of the laboratory from the back hole, and $T$ the hold time.  The gravitational analogue replaces the electromagnetic dipole or current source with the corresponding quadrupolar source. This scaling determines the physical regime accurately reproduced by the channel.

Furthermore, examples involving Killing horizons establish the appropriate observer-independent terminology. The designation ``soft'' is contingent upon the chosen flow and detector algebra: a mode classified as soft with respect to a horizon generator may be classified as hard with respect to inertial time. Consequently, the channel formalism employs the distinguishability induced on a specified algebra, in conjunction with its detector response, rather than relying on an invariant coordinate-frequency classification.

The local-noise formulation yields an identical positive kernel when expressed in laboratory variables. In QED, the decoherence functional is evaluated using the two-point function within the local laboratory frame. Let the difference in conserved currents be defined as
\begin{equation}
       f^\mu(x)=J_a^\mu(x)-J_b^\mu(x).
\end{equation}
Then the entangling photon number is
\begin{equation}
       \langle N\rangle_{ab}
       =
       \int \dd^4x\,\dd^4x'\,
       f_\mu(x)\,
       N^{\mu\nu}(x,x')\,
       f_\nu(x'),
       \label{eq:local_noise_general}
\end{equation}
where $N^{\mu\nu}$ denotes the symmetrized electromagnetic two-point kernel, with the gauge-dependent longitudinal sector eliminated either by current conservation or via gauge-invariant field-strength smearing. Consequently, the horizon representation and the local representation correspond to the identical kernel formulated on distinct algebras.

The consideration of finite-time recovery introduces a rigorous operational framework.  A horizon cut defines a restricted radiation state, whereas the final visibility depends on both the complete source protocol and the recovery map optimized over the accessible exterior output.  Quantum channel theory provides the requisite formal object: the reduced charged state emerges as the output of a map characterized by accessible and complementary radiation records.

The initial branch state serves as a probe for this map.  A two-branch spatial superposition extracts a single off-diagonal channel coefficient; it does not, however, define the channel in its entirety. In the soft-QED derivation, the assumption of an initial product state decouples the dressed charged branch sector from the independent radiation field, allowing the initial branch state $\rho_{\rm br}^{\rm in}$ to be diagonal, mixed, or coherent. A diagonal branch mixture undergoes radiation but lacks the initial coherence required for interference. Conversely, a coherent branch input reveals the off-diagonal action of the identical completely positive channel.

\section{Dressed charged branches and exterior algebra}

Consider a controlled charged system possessing an orthogonal branch basis $\{|a\rangle\}_{a\in\mathcal S}$, wherein each branch corresponds semiclassically to a conserved current $J_a^\mu(x)$.  The branch label denotes either the two arms of a Stern--Gerlach apparatus or a finite set of wavepackets within a larger interferometer. At leading order, the current is treated as classical; quantum fluctuations of internal degrees of freedom are incorporated by replacing $J_a$ with branch-conditioned current operators.

A physical charged branch does not correspond to a bare matter state; rather, it is a gauge-invariant dressed state given by
\begin{equation}
       |a\rangle_{\dress}
       =
       W_{\Coul}[J_a]\,|a\rangle_{\matter}\otimes|0\rangle_{\rad},
       \label{eq:dressed_branch}
\end{equation}
where $W_{\Coul}[J_a]$ denotes a Coulombic dressing determined by Gauss's law and the chosen asymptotic dressing convention. While the exact dressing is not unique, this nonuniqueness is physically significant only insofar as it modifies the radiative soft content. Pure Coulomb distinguishability does not constitute genuine decoherence, representing the QED analogue of the distinction between true and spurious decoherence.

In the exterior region of a black hole, the late-time radiation algebra comprises at least two relevant components,
\begin{equation}
       \A_{\rad}^{\late}
       \simeq
       \A_{\Iplus}\vee\A_{\Hplus},
       \label{eq:rad_split}
\end{equation}
subject to standard gauge constraints, edge-mode subtleties, and dressing conventions.  The symbol $\vee$ denotes the algebra generated by the two commuting asymptotic subalgebras.  The exterior observer measures a subalgebra of $\A_{\Iplus}$, leaves the remaining exterior output unresolved, and employs classical feed-forward in a recovery protocol.  Crucially, the observer lacks direct access to $\A_{\Hplus}$ via exterior operations.

Consequently, the accessible algebra does not merely encompass ``all photons exterior to the horizon'' in a naive spacetime sense. Rather, it is a specified detector algebra,
\begin{equation}
       \A_{\det}\subseteq\A_{\Iplus}\vee\A_{\lab},
       \label{eq:det_alg}
\end{equation}
characterized by finite temporal, angular, frequency, and polarization resolutions. The complement relevant to decoherence is therefore not simply ``the black hole,'' but instead
\begin{equation}
       \A_{\env}
       =
       \A_{\Hplus}\vee
       \A_{\Iplus\setminus\det}\vee
       \A_{\hard,\unres}\vee\cdots .
       \label{eq:env_alg}
\end{equation}
Only upon the specification of this algebraic bipartition does the discussion of decoherence become operationally meaningful.

Formulating the horizon sector as a quantum channel algebra also bridges the present soft construction with earlier information-theoretic perspectives on black holes. Bekenstein's arguments concerning entropy and information flow, and in particular the observation by Bekenstein and Mayo that black holes behave in several respects as one-dimensional communication channels, are not utilized here to derive the soft-photon Gram factor microscopically~\cite{BekensteinInfoReview,BekensteinEntropyFlow,BekensteinMayo}. Nevertheless, they provide the appropriate operational intuition. The horizon is not modeled as an opaque measuring apparatus; rather, within the exterior problem, it constitutes a unidirectional complementary output algebra, resolved here exclusively in its soft radiative sector. The microscopic formulation thus yields a more precise account than the thermodynamic analogy: for fixed branch currents, the horizon soft channel generates a Gram matrix of branch-conditioned records, and this Gram matrix acts as a Schur multiplier on the charged-branch coherences.

The following Figure 1 and 2 serve solely as schematic aids for this algebraic bipartition. They are included in this section, rather than in the introduction, because their function is strictly technical: to delineate the accessible exterior record from the complementary horizon output.

\begin{figure}[t]
\centering
\begin{tikzpicture}[scale=1.2, every node/.style={font=\small}, >=Stealth]
  \coordinate (H0) at (0.2,0.5);
  \coordinate (Ip) at (2.8,5.4);
  \coordinate (I0) at (5.4,0.5);

  \fill[gray!12] (-0.65,0.5) -- (H0) -- (Ip) -- (1.55,5.4) -- cycle;
  \node[rotate=62,gray!100] at (0.75,3.25) {black-hole interior};

  \draw[thick] (H0) -- (Ip) node[pos=0.45,above left,sloped] {$\Hplus$};
  \draw[thick] (I0) -- (Ip) node[pos=0.4,above right,sloped] {$\Iplus$};
  \draw[thick] (H0) .. controls (1.75,0.02) and (3.85,0.02) .. (I0)
       node[pos=0.52,below] {$\Sigma_{\rm in}$};
  \node[above] at (Ip) {$i^+$};
  \node[below left] at (H0) {$\mathcal B$};
  \node[below right] at (I0) {$i^0$};

  \draw[densely dotted,rounded corners] (2.55,1.18) rectangle (3.65,4.35);
  \node[above] at (3.10,4.2) {Alice lab};
  \draw[very thick] (2.86,1.30) .. controls (2.62,2.15) and (2.62,3.05) .. (2.93,4.10)
       node[pos=0.6,left] {$J_a$};
  \draw[very thick,dashed] (3.34,1.30) .. controls (3.62,2.15) and (3.62,3.05) .. (3.27,4.10)
       node[pos=0.45,right] {$J_b$};

  \draw[->,decorate,decoration={snake,amplitude=0.65mm,segment length=3.0mm}]
       (3.53,2.6) .. controls (4.18,2.95) and (4.55,3.65) .. (4.43,4.22);
  \node[align=center,right] at (4.6,4.12) {$\Phi^{\Iplus}_{a,b}$\\exterior records};

  \draw[->,decorate,decoration={snake,amplitude=0.65mm,segment length=3.0mm}]
       (2.78,2.32) .. controls (2.00,2.75) and (1.45,3.22) .. (1.16,3.42);
  \draw[->,decorate,decoration={snake,amplitude=0.65mm,segment length=3.0mm}]
       (2.78,3.18) .. controls (2.23,3.50) and (1.82,3.95) .. (1.54,4.14);
  \node[align=center,left] at (0.75,4.1) {$\Phi^{\hor}_{a,b}$\\horizon records};

  \node[draw,rounded corners,align=center,fill=white] at (7.5,2.45)
      {$\A_{\det}\subset\A_{\Iplus}$\\measure / ignore / condition};
  \node[draw,rounded corners,align=center,fill=white] at (3.1,0.72)
      {$\rho^{\rm in}_{\rm br}\otimes\rho^{\rm in}_{\rm rad}$};
  \node[draw,rounded corners,align=center,fill=white] at (3.45,6.18)
      {$\E_{\rm out}=\Tr_{\Hplus}\Tr_{\Iplus_{\rm unres}} \Ad_{S_{\rm dress}}$};

  \draw[->] (3.45,5.92) -- (3.18,5.03);
  \draw[->] (5.48,2.92) -- (4.58,3.52);
  \node[align=center,gray!110] at (-0.45,1.25) {inaccessible\\complement};
\end{tikzpicture}
\caption{Penrose-diagram illustrating the horizon-restricted open-system channel.  The branch currents $J_a$ and $J_b$ are localized within Alice's exterior laboratory and act as sources for branch-dependent radiative records.  Records reaching $\Iplus$ constitutes, subject to detector resolution, the exterior measurement algebra; these records are subsequently measured, disregarded, or employed in a feed-forward recovery protocol. Conversely, radiation traversing $\Hplus$ correspond to the complementary channel relative to the exterior observer, which modulates the off-diagonal element of the charged sector by the horizon Gram factor $G^{\hor}_{ab}=\langle\Phi_b^{\hor}|\Phi_a^{\hor}\rangle$. Consequently, the diagram does not depict a physical detector array distributed throughout spacetime; rather, it provides a conformal representation of the algebraic bipartition that renders standard exterior Bloch--Nordsieck inclusivity incomplete.}
\label{fig:penrose_channel}
\end{figure}
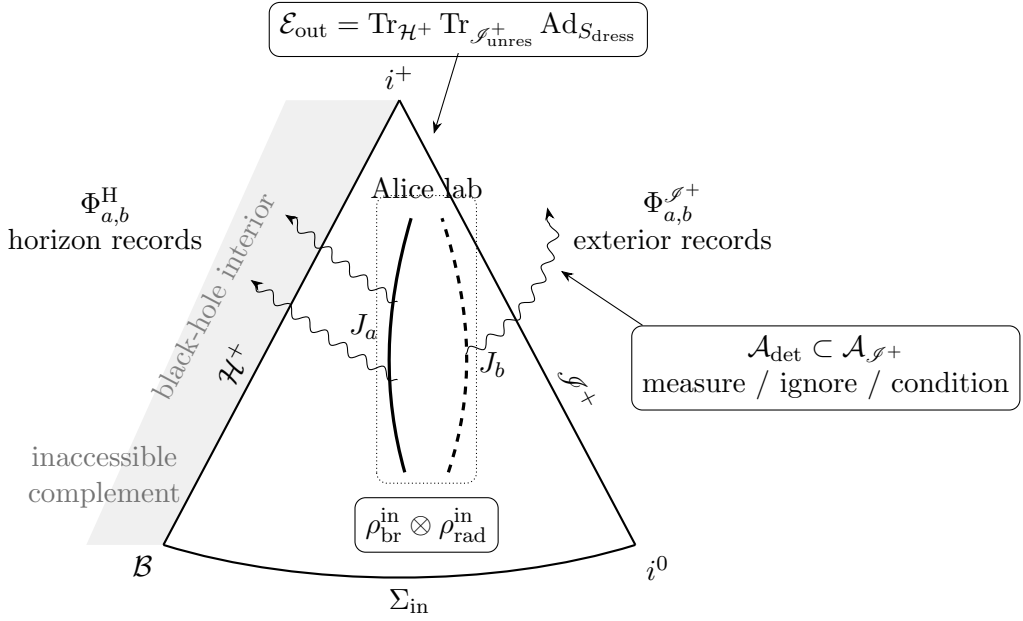

\begin{figure}[t]
\centering
\begin{tikzpicture}[scale=0.90, every node/.style={font=\small}, >=Stealth]
  \coordinate (UL) at (-2.35,2.35);
  \coordinate (UR) at ( 2.35,2.35);
  \coordinate (R)  at ( 4.75,0.00);
  \coordinate (DR) at ( 2.35,-2.35);
  \coordinate (DL) at (-2.35,-2.35);
  \coordinate (L)  at (-4.75,0.00);
  \coordinate (O)  at (0,0);

  \fill[blue!7] (O)--(UR)--(R)--(DR)--cycle;

  \draw[thick] (UL)--(UR)--(R)--(DR)--(DL)--(L)--cycle;

  \draw[very thick] (UL)--(UR);
  \draw[very thick] (DL)--(DR);
  \node at (0,2.68) {$r=0$};
  \node at (0,-2.68) {$r=0$};

  \draw[thick] (O)--(UR);
  \draw[thick] (O)--(UL);
  \draw[thick] (O)--(DR);
  \draw[thick] (O)--(DL);

  \node[above right] at (3.7,0.6) {$\mathscr I^+_R$};
  \node[right] at (5.00,0.00) {right exterior};
  \node[rotate=45] at (0.5,0.9) {$\mathcal H^+_R$};

  \draw[densely dotted,rounded corners] (1.62,-0.42) rectangle (3.08,1.58);
  \draw[very thick] (2.06,-0.02) .. controls (1.90,0.48) and (1.98,0.98) .. (2.15,1.36);
  \draw[very thick,dashed] (2.50,-0.02) .. controls (2.70,0.48) and (2.66,0.98) .. (2.48,1.36);
  \node at (2.35,1.65) {Alice};

  \draw[->,decorate,decoration={snake,amplitude=0.50mm,segment length=2.7mm}]
      (2.66,0.78) .. controls (3.18,1.25) and (3.62,1.52) .. (4.0,1.86);
  \node[right] at (3.9,2.05) {$\Phi^{\mathscr I}_{a,b}$};

  \draw[->,decorate,decoration={snake,amplitude=0.50mm,segment length=2.7mm}]
      (2.04,0.84) .. controls (1.50,1.10) and (1.02,1.36) .. (0.66,1.70);
  \node[left] at (0.65,1.85) {$\Phi^{H}_{a,b}$};
\end{tikzpicture}
\caption{Planar $(r,t)$ Penrose diagram of the fully extended eternal Schwarzschild geometry, drawn in the conventional hexagon form.  The thick upper and lower horizontal segments represent the spacelike singularities $r=0$.  The diagonal line labelled $\mathcal H^+_R$ is the future event horizon bounding the right exterior region occupied by Alice.  Branch-dependent radiative records reaching $\mathscr I^+_R$ are, up to detector resolution, accessible exterior records, whereas branch-dependent records crossing $\mathcal H^+_R$ are assigned to the complementary horizon channel.  The left exterior, the white-hole region, and most other labels are omitted because they are spectators for the open-system construction.}
\label{fig:eternal_planar_penrose}
\end{figure}
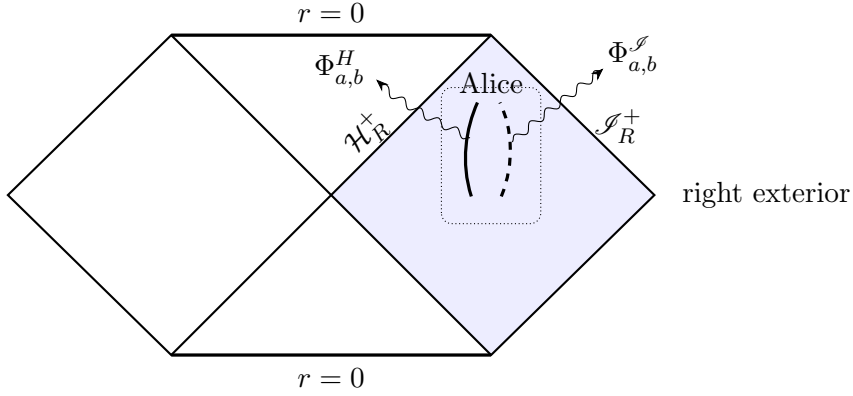

\section{Factorized initial state and branch coherence} \label{sec:product_branch_coherence}

The preceding section established the physical branch basis and the algebraic bipartition. I now isolate a conceptual subtlety that is often obscured by standard notation but is crucial for connecting ordinary soft-QED scattering to the horizon-decoherence interferometer. The product initial state employed in the soft-QED open quantum system (OQS) construction does not assert that the local charged system is initialized in a single classical configuration. Rather, it asserts that, once the physical dressing convention has been fixed, the independent radiative degrees of freedom are initially unentangled with the controlled charged branch space:
\begin{equation}
       \rho_{\tot}^{\inm}
       =
       \rho_{\rm br}^{\inm}\otimes\rho_{\rad}^{\inm} .
       \label{eq:general_product_initiality}
\end{equation}
The branch density matrix is otherwise arbitrary,
\begin{equation}
       \rho_{\rm br}^{\inm}
       =
       \sum_{a,b\in\mathcal S}\rho^{\inm}_{ab}|a\rangle\langle b| .
       \label{eq:arbitrary_branch_density}
\end{equation}
A diagonal \(\rho_{\rm br}^{\inm}\) corresponds to an incoherent preparation, which is sufficient for computing inclusive emission probabilities or for identifying the branch-conditioned radiation produced by a classical current. However, such a state cannot exhibit interference loss, as no off-diagonal coherences are initially present. Conversely, an interferometric experiment requires a coherent branch input, such as
\begin{equation}
       |\psi\rangle
       =c_1|1\rangle+c_2|2\rangle,
       \qquad
       \rho_{12}^{\inm}=c_1c_2^* ,
       \label{eq:two_branch_probe_input}
\end{equation}
wherein the output visibility is determined by the coefficient multiplying \(\rho_{12}^{\inm}\).

This distinction establishes the logical structure of the analysis. A single-branch product input defines a single column of the Stinespring isometry, as it determines the radiation state generated by the current \(J_a\).  Linearity then dictates the action on any coherent superposition of branches. Consequently, the specific initial state is incidental to the structural derivation: it merely represents a choice of branch basis and input vector used to characterize an already-defined channel. The product-state scattering setup and the interferometric superposition setup are not mutually exclusive assumptions; rather, they address distinct questions. The former derives the channel, while the latter probes its off-diagonal action.

In the coherent-source approximation, the branch-controlled isometry takes the form
\begin{equation}
       V:
       |a\rangle\otimes|\Omega\rangle
       \longmapsto
       |a\rangle_{\dress}^{\out}
       \otimes
       |\Gamma_a^{\Iplus}\rangle
       \otimes
       |\Gamma_a^{\hor}\rangle .
       \label{eq:branch_isometry_product_section}
\end{equation}
Applying \(V\) to Eq.~\eqref{eq:arbitrary_branch_density} and tracing out the environmental degrees of freedom induces a Schur multiplier action on the matrix elements. Therefore, the branch superposition does not constitute an additional dynamical postulate; rather, it is the input state upon which the previously derived open-system map becomes experimentally observable. In the language of quantum information theory, it represents the choice of a code basis and a probe state, rather than a novel dynamical mechanism. The channel formalism is thus strictly more general than the two-branch overlap framework, as it encompasses classical emission, incoherent mixtures, two-arm interferometers, and multi-branch coherent code subspaces within a single unified expression.
A crucial qualification concerns the tensor factorization. Equation~\eqref{eq:general_product_initiality} is defined only after gauge dressing, rather than with respect to bare charged matter and a bare photon Fock vacuum, as Gauss's law necessitates such a dressing. The precise formulation requires product initiality between the dressed controlled branch sector and the independent radiative sector. Failure to observe this distinction would reintroduce spurious Coulomb decoherence as an artifact of an unphysical tensor factorization. With this qualification, the connection to the prior soft-QED scattering analysis is direct: the product state serves as the natural input for constructing the channel map \(\E\), whereas the black-hole application permits \(\rho_{\rm br}^{\inm}\) to possess arbitrary coherences, subsequently analyzing how \(\E\) suppresses them.

Finally, a remark regarding operational meaning is warranted.
If \(\rho_{ab}^{\inm}=0\), no experiment can exhibit decoherence in the \((a,b)\) sector, as no such coherence is initially present.  Nevertheless the channel coefficient \(G_{ab}\) remains well-defined and can be extracted by preparing a suitable coherent input or by reconstructing the process tensor. Consequently, the standard two-branch spatial superposition serves merely as an operational probe of a channel coefficient, rather than a structural assumption or an alternative dynamical framework.

\section{From the flat-space soft-QED channel to a horizon output}

The soft-QED OQS analysis of \cite{Rey2026a,Rey2026b} provides the foundation employed in this work; however, it should be adopted with modification.  In flat spacetime, or within a finite cavity serving as an infrared regulator, the hard-sector channel has a simple eikonal form: populations in the hard basis are conserved, whereas coherences are multiplied by an influence factor,
\begin{equation}
       \rho_{pp'}(t)
       =
       \rho_{pp'}(0)\exp[-\Gamma_R(p,p',t)-\ii\Phi_R(p,p',t)] .
       \label{eq:paperA_channel_template}
\end{equation}
The central input is the bath spectral density, which, in cavity notation, can be expressed schematically as
\begin{equation}
       \Gamma_R(p,p',t)
       =
       \frac{1}{\pi}\int_0^\infty d\omega\,
       \frac{J_R(\omega)}{\omega^2}
       (1-\cos\omega t)\,\Delta\mathcal J_{pp'}(\omega),
       \label{eq:paperA_spectral_template}
\end{equation}
wherein the identical kernel governs the Sudakov factor, hard-sector decoherence, purity loss, and the asymptotic approach to a dephasing channel.  The diagonal identity \(\Gamma_R(p,p,t)=0\) constitutes the OQS manifestation of inclusive probability conservation; the off-diagonal exponent corresponds to the cloud-distinguishability functional. A finite cavity size introduces a gap in the bath spectrum and permits recurrences, whereas the infinite-volume or long-observation limit restores the logarithmic infrared memory. These features constitute the precise conceptual content from the flat-space OQS problem that is retained in the present analysis.

The black-hole scenario is not derived merely by substituting the cavity radius with a horizon radius, nor by transplanting the unparticle interpretation of the electron. Rather, the inheritance is structural rather than literal.  The hard momentum labels \(p,p'\) are superseded by branch currents \(J_a,J_b\); the scalar spectral density is generalized to a positive noise kernel acting on current differences; and the infrared regulator is replaced by a detector algebra coupled with a complementary horizon output. Consequently, the relevant exponent assumes the form
\begin{equation}
       \Gamma^{\rm comp}_{ab}
       =
       \frac12\int \dd^4x\,\dd^4x'\,
       f_{ab,\mu}(x)\,
       N_{\rm comp}^{\mu\nu}(x,x')\,
       f_{ab,\nu}(x'),
       \qquad
       f_{ab}=J_a-J_b,
       \label{eq:comp_kernel_from_paperA}
\end{equation}
where \(N_{\rm comp}\) denotes the component of the radiative noise kernel assigned to inaccessible or unconditioned outputs, encompassing the horizon sector and any unresolved exterior sector. Within a factorized coherent regime, \(N_{\rm comp}\) contains a horizon block, and Eq.~\eqref{eq:comp_kernel_from_paperA} reduces to the horizon Gram factor utilized below.  In the presence of cross correlations, the same equation is interpreted as a complementary-channel kernel rather than a product of independent \(\Iplus\) and \(\Hplus\) factors.

This correspondence yields two essential safeguards. First, the fixed-history Feynman--Vernon normalization remains identical to its formulation in the flat-space channel: \(\F[J,J]=1\).  The horizon does not induce a diagonal infrared anomaly. Second, the source of decoherence is once again the off-diagonal distinguishability of environmental records, now refined by the algebraic distinction that certain records are exterior and controllable, whereas horizon records constitute complementary outputs for the semiclassical exterior observer. Consequently, the novel aspect of the black-hole scenario is not the mere existence of soft-induced dephasing. Rather, it is the causal and algebraic bipartition of the infrared channel into accessible and complementary outputs, coupled with the multi-branch Gram geometry that emerges from treating this bipartition as a completely positive map.

\section{Leading soft input and horizon projection}\label{sec:leading_soft_input}

The preceding analysis establishes the channel structure inherited from the flat-space soft-QED OQS framework. I now specify the soft-theorem input employed in this work, restricting our attention exclusively to the leading, eikonal photon theorem. For a soft photon with momentum $q^\mu=\omega \hat q^\mu$, helicity $\lambda$, and $\omega\to0$, the leading factor multiplying a hard charged process is~\cite{WeinbergSoft}
\begin{equation}
       S^{(0)}_\lambda(q)
       =
       e\sum_i\eta_i Q_i
       \frac{p_i\cdot\varepsilon_\lambda(q)}{p_i\cdot q} .
       \label{eq:leading_soft_factor}
\end{equation}
Here $Q_i$ denotes the charge in units of $e$, $\eta_i=+1$ for outgoing charged legs and $\eta_i=-1$ for incoming charged legs, and the polarization vector is transverse.  Equation~\eqref{eq:leading_soft_factor} represents the amplitude-level manifestation of the radiative displacement generated by a conserved eikonal current. The angular-momentum operator $J_i^{\mu\nu}$ does not enter the present analysis; its inclusion pertains to the subleading Low--Burnett--Kroll sector, which is reserved for a subsequent study.

Let $\K_\Omega$ denote the background- and state-dependent map from a conserved branch current to the physical one-particle radiative displacement in the chosen semiclassical state $\Omega$.  In flat spacetime, this map corresponds to the standard on-shell current-to-photon map underlying the leading soft theorem; in the black-hole exterior, it additionally incorporates propagation through the curved geometry and the specification of the out-algebra. For a given branch  $a$, I define
\begin{equation}
       \alpha_a^{(0)}=\K_\Omega J_a,
       \qquad
       \alpha_a^{H,(0)}=P_H\K_\Omega J_a,
       \label{eq:leading_displacement_projection}
\end{equation}
where $P_H$ denotes the projection of the radiative displacement onto the horizon output algebra. Consequently, the branch-pair data entering the exterior channel is given by
\begin{equation}
       \Delta\alpha_{ab}^{H,(0)}
       =
       \alpha_a^{H,(0)}-\alpha_b^{H,(0)}
       =
       P_H\K_\Omega(J_a-J_b).
       \label{eq:leading_delta_alpha_projection}
\end{equation}
The leading-order horizon Gram matrix element is given by 
\begin{equation}
       G_{ab}^{H,(0)}
       =
       \exp\left[
          -\frac12\|\Delta\alpha_{ab}^{H,(0)}\|^2
          +\ii\,\operatorname{Im}
          \langle\alpha_b^{H,(0)},\alpha_a^{H,(0)}\rangle
       \right]
       \equiv
       \exp[-\Gamma_{ab}^{H,(0)}+\ii\Phi_{ab}^{H,(0)}] .
       \label{eq:leading_horizon_gram_explicit}
\end{equation}
Equivalently, expressed in Feynman--Vernon notation,
\begin{equation}
       \Gamma_{ab}^{H,(0)}
       =
       \frac12
       \int (J_a-J_b)_\mu
       N_H^{(0)\mu\nu}
       (J_a-J_b)_\nu,
       \label{eq:leading_horizon_noise_kernel}
\end{equation}
where $N_H^{(0)}$ denotes the leading soft radiative noise kernel assigned to the horizon complementary output. The superscript 
$(0)$ serves as a bookkeeping device to isolate the leading/eikonal channel, thereby preventing conflation with the subleading contributions addressed in the sequel. In subsequent sections, provided no ambiguity arises, I suppress the superscript $(0)$ and denote the leading horizon quantities simply as $G^H_{ab}$, $\Gamma^H_{ab}$, and $\Phi^H_{ab}$.

This section also delineates the scope of the claims advanced in this work. I do not posit a universal subleading horizon theorem, nor do I employ the Low--Burnett--Kroll angular-momentum soft operator. The precise claim is that, whenever the semiclassical horizon output admits the projection~\eqref{eq:leading_displacement_projection} of the leading eikonal displacement, the exterior channel incorporates the leading horizon Schur factor~\eqref{eq:leading_horizon_gram_explicit}. Subsequent sections analyze the implications of this leading complementary output, specifically examining complete positivity, Gram constraints, exterior eraser bounds, local-noise equivalence, and multi-branch interferometric tests.

\section{The leading horizon Schur channel}

Building upon the conceptual distinction established in Sec.~\ref{sec:product_branch_coherence} and the leading-order projection detailed in Sec.~\ref{sec:leading_soft_input}, I first consider the scenario in which the branch currents source coherent radiation states. The dressed scattering map assumes the branch-controlled form
\begin{equation}
       U_{\dress}:
       |a\rangle_{\dress}\otimes|\Omega\rangle
       \longmapsto
       |a\rangle_{\dress}^{\out}
       \otimes
       |\Gamma_a^{\Iplus}\rangle
       \otimes
       |\Gamma_a^{\hor}\rangle,
       \label{eq:branch_unitary}
\end{equation}
where $|\Omega\rangle$ denotes the chosen initial state of the field, such as the Unruh vacuum, and $|\Gamma_a^{\Iplus}\rangle$, $|\Gamma_a^{\hor}\rangle$ are coherent states of radiative data on null infinity and the horizon.  A general input charged density matrix
\begin{equation}
       \rho=\sum_{a,b}\rho_{ab}|a\rangle\langle b|
\end{equation}
is mapped, after tracing over the horizon degrees of freedom, to
\begin{equation}
       \rho^{\Iplus}_{ab}
       =
       \rho_{ab}\,
       |a,\Gamma_a^{\Iplus}\rangle
       \langle b,\Gamma_b^{\Iplus}|
       \;G^{\hor}_{ab},
       \qquad
       G^{\hor}_{ab}=\langle\Gamma_b^{\hor}|\Gamma_a^{\hor}\rangle .
       \label{eq:horizon_trace}
\end{equation}
If the exterior radiation is also traced out, the reduced charged density matrix obeys
\begin{equation}
       (\E_{\rad}\rho)_{ab}
       =
       G^{\Iplus}_{ab}G^{\hor}_{ab}\rho_{ab},
       \qquad
       G^{\Iplus}_{ab}=\langle\Gamma_b^{\Iplus}|\Gamma_a^{\Iplus}\rangle .
       \label{eq:full_rad_schur}
\end{equation}
Consequently, the black-hole contribution manifests as the Schur multiplier
\begin{equation}
       (\E_{\hor}\rho)_{ab}=G^{\hor}_{ab}\rho_{ab} .
       \label{eq:hor_schur}
\end{equation}

\begin{proposition}[Gram positivity of the horizon channel]
Let $\{|\Gamma_a^{\hor}\rangle\}_{a\in\mathcal S}$ denotes a set of normalized vectors in the horizon radiation Hilbert space, and define the Gram matrix elements as $G^{\hor}_{ab}=\langle\Gamma_b^{\hor}|\Gamma_a^{\hor}\rangle$.  The map $\rho\mapsto G^{\hor}\circ\rho$, where $\circ$ denotes the Hadamard (entrywise) multiplication in the branch basis, is completely-positive and trace-preserving.
\end{proposition}

\begin{proof}
For any complex coefficients $c_a$,
\begin{equation}
       \sum_{a,b}c_a^*G^{\hor}_{ab}c_b
       =
       \left\|\sum_a c_a|\Gamma_a^{\hor}\rangle\right\|^2\ge0.
\end{equation}
Thus, $G^{\hor}$ is positive semidefinite. Given that $G^{\hor}_{aa}=1$, the Schur product map preserves the trace.  Complete positivity follows directly from the Schur product theorem, or can be demonstrated explicitly by diagonalizing $G^{\hor}=\sum_r\lambda_r v^{(r)}v^{(r)\dagger}$ and constructing the Kraus operators
\begin{equation}
       K_r=\sqrt{\lambda_r}\sum_a v^{(r)}_a |a\rangle\langle a|,
       \qquad
       \E_{\hor}(\rho)=\sum_r K_r\rho K_r^\dagger .
\end{equation}
The normalization condition $G^{\hor}_{aa}=1$ ensures that $\sum_rK_r^\dagger K_r=\mathbf 1$.
\end{proof}

the horizon-overlap framework to the status of a rigorous quantum channel. Furthermore, it imposes nontrivial consistency conditions across multiple branches. For systems comprising three or more branches, the pairwise decoherence factors cannot be assigned arbitrarily; rather, they must collectively constitute a positive semidefinite Gram matrix. This provides a rigorous, falsifiable constraint on any proposed phenomenological model of horizon-induced decoherence.

In the case of coherent radiation, the Gram entries take the form
\begin{equation}
       G^{\hor}_{ab}
       =
       \exp\left[
          -\half\|\alpha_a^{\hor}-\alpha_b^{\hor}\|^2
          +\ii\,\operatorname{Im}\langle\alpha_b^{\hor},\alpha_a^{\hor}\rangle
       \right].
       \label{eq:coherent_gram}
\end{equation}
The real part governs the decoherence, while the imaginary part corresponds to a branch-dependent radiative phase. An analogous structure holds for gravitational coherent states upon substituting the electromagnetic current with the stress-energy tensor and the vector potential with the linearized metric perturbation, subject to the center-of-mass and gauge constraints pertinent to the gravitational sector.

\section{Leading Feynman-Vernon horizon inclusivity}

The coherent-state formula constitutes the pure-state realization of the Feynman--Vernon influence functional. For conserved currents within a Gaussian electromagnetic state, the reduced charged dynamics assumes the closed-time-path (CTP) form
\begin{equation}
       \F[J_+,J_-]
       =
       \exp\left\{
          \ii\int J^-_\mu D_{\mathrm R}^{\mu\nu}J^+_\nu
          -\half\int J^-_\mu N^{\mu\nu}J^-_\nu
       \right\},
       \label{eq:FV}
\end{equation}
where $J^\pm=(J_+\pm J_-)/2$, $D_{\mathrm R}$ denotes the retarded propagator, and $N$ represents the Hadamard (noise) kernel of the radiative sector retained in the trace.  For branch labels $a,b$ one identifies $J_+=J_a$, $J_-=J_b$, such that the decoherence exponent is governed by the quadratic form in the current difference $J_a-J_b$.

The primary identity is that of normalization. When the two histories coincide,
\begin{equation}
       {
       \F[J,J]
       =
       \Tr_{\rad}[U[J]\rho_{\rad}U[J]^\dagger]
       =1 .}
       \label{eq:fixed_history_FV_identity}
\end{equation}
This represents the Feynman--Vernon manifestation of fixed-outcome unitarity and the real/virtual infrared cancellation inherent in Bloch--Nordsieck reasoning. The horizon effect discussed in this work does not violate Eq.~\eqref{eq:fixed_history_FV_identity}; it does not constitute a residual diagonal infrared divergence in a reduced hard probability. Rather, it pertains to unequal histories,
\begin{equation}
       G_{ab}
       =
       \F[J_a,J_b],
       \qquad J_a\ne J_b,
       \label{eq:G_full_FV}
\end{equation}
which modulate the off-diagonal elements of the branch or hard-sector density matrix.

In favorable semiclassical scattering regimes, the late-time radiation algebra is represented by exterior modes at $\Iplus$ and horizon-crossing modes at $\Hplus$.  It is crucial, however, not to elevate this geometric statement to the status of an unconditional factorization theorem. The general Gaussian kernel, following an $\Iplus/H$ bipartition, assumes the form
\begin{equation}
       N
       =
       N_{\Iplus\Iplus}
       +N_{HH}
       +N_{\Iplus H}
       +N_{H\Iplus},
       \label{eq:general_cross_kernel_split}
\end{equation}
with an analogous structure for the dissipative phase kernel. Consequently, the decoherence exponent comprises
\begin{equation}
       \Gamma_{ab}
       =
       \Gamma^{\Iplus}_{ab}
       +
       \Gamma^{H}_{ab}
       +
       \Gamma^{\Iplus H}_{ab} .
       \label{eq:Gamma_cross_terms}
\end{equation}
The cross term vanishes only when the chosen out-mode representation, covariance, and state render the exterior and horizon radiative sectors orthogonal, or block-diagonal, with respect to the family of branch-current differences under consideration.

Thus, the product formula
\begin{equation}
       G_{ab}
       =
       G^{\Iplus}_{ab}G^H_{ab}
       \label{eq:factorized_G_conditioned}
\end{equation}
is not a kinematic identity but rather characterizes a controlled factorized-channel regime. A sufficient condition for this is a one-particle direct sum
\begin{equation}
       \K_{\rad}
       =
       \K_{\Iplus}\oplus\K_H,
       \qquad
       \alpha_a
       =
       \alpha_a^{\Iplus}\oplus\alpha_a^H,
       \label{eq:one_particle_direct_sum}
\end{equation}
for the branch-dependent coherent displacement.  Then coherent-state inner products then yields
\begin{equation}
       G_{ab}
       =
       \exp\left[-\half\|\alpha_a^{\Iplus}-\alpha_b^{\Iplus}\|^2
              -\half\|\alpha_a^{H}-\alpha_b^{H}\|^2
              +\ii\Phi_{ab}\right]
       =
       G^{\Iplus}_{ab}G^H_{ab}.
       \label{eq:coherent_product_G}
\end{equation}
In this regime the horizon contribution is given by
\begin{equation}
       G^{H,(0)}_{ab}
       =
       \exp\left[
          \ii\Phi^{H,(0)}_{ab}
          -\half\int (J_a-J_b)_\mu
                   N_{H}^{(0)\mu\nu}
                   (J_a-J_b)_\nu
       \right],
       \label{eq:horizon_FV_gram_revised}
\end{equation}
with $N_H$  is positive definite on physical radiative test currents. In the presence of cross kernels, $G^H_{ab}$ is replaced by the appropriate complementary-channel distinguishability or conditional fidelity; the channel formalism remains valid, whereas the naive product notation does not.
This refined bookkeeping also clarifies the relationship to black-hole complementarity. I neither assume nor deny black-hole complementarity. If an exact theory of quantum gravity renders the complete asymptotic algebra at $\Iplus$ information-complete, this would constitute a statement regarding $\A_{\Iplus}^{\rm exact}$, rather than an identification of the ordinary semiclassical soft algebra $\A_{\Iplus}^{\rm semi,soft}$ with the horizon soft algebra $\A_H^{\rm semi,soft}$. The present construction is strictly a semiclassical exterior-channel statement: within the low-energy radiative algebra employed in the Feynman--Vernon calculation, the horizon output acts as a complementary environment unless an explicit reconstruction map is provided.
Flat-space infrared cancellation can now be articulated within the channel formalism. In standard scattering processes, hard charged observables are not described by exclusive amplitudes with a fixed number of zero-energy photons; rather, they are described by inclusive probabilities, or equivalently, by a reduced hard channel obtained after tracing over unresolved soft radiation up to a detector resolution 
$\Lambda$. Virtual infrared divergences and real unresolved radiation combine to yield a finite, completely positive map on the hard charged sectors. In the open quantum system (OQS) description, the Bloch--Nordsieck factor corresponds to the influence functional of the soft bath, and Eq.~\eqref{eq:fixed_history_FV_identity} represents the trace-preserving normalization for a fixed hard history.

The black-hole exterior modifies the unequal-history sector. An exterior observer can trace, measure, or condition upon radiation reaching $\Iplus$, but cannot employ standard exterior operations to erase a complementary horizon record. Consequently, the residual black-hole effect does not constitute a violation of fixed-history Feynman--Vernon unitarity; rather, it represents an obstruction to exterior control over the unequal-history infrared factor.

This constitutes the central conceptual result of this work: {\sl black-hole leading-soft decoherence is the same as inaccessible off-diagonal Feynman-Vernon record}.
This statement is intended in an operational sense. It holds once a detector algebra, a resolution function, a dressing convention, and, where applicable, a factorized out-channel representation have been specified.

\begin{theorem}[Leading horizon-restricted off-diagonal Feynman-Vernon channel] Fix a dressed charged branch subspace, a detector resolution function 
$\chi$ for the exterior radiation, and a late-time radiation representation. The complete normalized influence functional obeys  $\F[J,J]=1$.  For unequal branch histories $J_a\ne J_b$, the exterior reduced branch channel assumes the Schur form
\begin{equation}
       (\E_{\rm out}\rho)_{ab}
       =
       \rho_{ab}G^{\rm out}_{ab},
       \qquad
       G^{\rm out}_{ab}=\F_{\rm out}[J_a,J_b].
       \label{eq:restricted_FV_theorem_general}
\end{equation}
In a factorized $\Iplus/H$ out-channel regime,
\begin{equation}
       G^{\rm out}_{ab}
       =
       G^{\Iplus}_{ab}(\chi)G^H_{ab},
       \label{eq:restricted_FV_theorem_factorized}
\end{equation}
where $G^{\Iplus}_{ab}$ is dependent on the detector and conditioning, whereas $G^H_{ab}$ is determined by the complementary horizon output. Consequently, maximal exterior infrared inclusivity can eliminate or condition upon exterior records but cannot, within the semiclassical exterior algebra, eliminate the horizon factor. \end{theorem}

\begin{proof}[Argument] The dressed evolution induces a Stinespring isometry from the branch space to the tensor product of the branch and radiation spaces. Fixed-history normalization is given by Eq.~\eqref{eq:fixed_history_FV_identity}. For unequal histories, tracing over or measuring a specified radiation algebra yields a Schur multiplier on the branch coherences. If the late-time radiation representation factorizes into exterior and horizon outputs, the coherent-state Gram inner product factorizes as shown in Eq.~\eqref{eq:coherent_product_G}. Exterior detector operations can modify the exterior factor, but not the complementary horizon factor. In the absence of such factorization, the same conclusion holds within the complementary-channel formalism, with the horizon contribution replaced by the corresponding conditional fidelity or distinguishability of the inaccessible records. \end{proof}

This theorem marks the first point at which the present analysis departs from a mere reformulation of the prior two-branch overlap picture. It predicts an operational residual: even an ideal exterior detector implementing maximal soft monitoring at $\Iplus$ yields a residual visibility bounded by the complementary horizon distinguishability,
\begin{equation}
       V_{\rm max}^{\rm ext}(a,b)
       \leq
       |G^{H}_{ab}|
       \label{eq:max_ext_visibility_revised}
\end{equation}
when the factorized horizon-record description is valid, and by the corresponding complementary-channel fidelity in the general case. This does not imply that photons are lost; rather, it asserts that exterior operations cannot control all unequal-history infrared records.

\begin{remark}[Scope beyond black holes] An analogous structure arises whenever spacetime geometry, boundary conditions, acceleration, or material response partitions the soft sector into detector-accessible and detector-inaccessible records. Black holes are distinguished by the fact that the inaccessible factor is geometrically sharp, and its low-frequency kernel is fixed by the horizon state. The channel distinction originates from this specific spectral and algebraic structure, rather than from the generic observation that black holes act as environments. \end{remark}

\begin{definition}[Triangular channel holonomy] For three branch histories with nonzero pairwise overlaps, define
\begin{equation}
       \mathcal B_{123}
       =
       G_{12}G_{23}G_{31},
       \qquad
       \mathcal U_{123}
       =
       {\mathcal B_{123}\over |\mathcal B_{123}|}
       =e^{\ii\Theta_{123}} .
       \label{eq:Bargmann_invariant}
\end{equation}
The complex number $\mathcal B_{123}$ is the Bargmann invariant of the environment-induced Gram matrix, and $\Theta_{123}$ is the corresponding channel holonomy phase.  Under rephasings of either the branch basis or the environmental representatives, the individual phases of $G_{ab}$ shift by a coboundary, whereas $\mathcal B_{123}$ remains invariant.  Consequently, the two-arm phase is convention-dependent, while the three-arm cyclic phase is not.
\end{definition}

For coherent leading-soft records,
\begin{equation}
       G_{ab}^{H,(0)}
       =
       \exp\left[-\half
       \|\alpha_a^{H,(0)}-\alpha_b^{H,(0)}\|^2
       +\ii\,\operatorname{Im}
       \langle\alpha_b^{H,(0)},\alpha_a^{H,(0)}\rangle\right].
       \label{eq:leading_coherent_pair_overlap}
\end{equation}
The horizon Bargmann invariant is therefore given by
\begin{equation}
       \mathcal B^{H,(0)}_{123}
       =
       \exp\left[-\half
       \sum_{\rm cyc}\|\alpha_a^{H,(0)}-\alpha_b^{H,(0)}\|^2
       +\ii\Theta^{H,(0)}_{123}\right],
       \label{eq:Bargmann_coherent}
\end{equation}
with
\begin{equation}
       \Theta^{H,(0)}_{123}
       =
       \operatorname{Im}\left(
       \langle\alpha_2^{H,(0)},\alpha_1^{H,(0)}\rangle
       +\langle\alpha_3^{H,(0)},\alpha_2^{H,(0)}\rangle
       +\langle\alpha_1^{H,(0)},\alpha_3^{H,(0)}\rangle
       \right).
       \label{eq:Theta_horizon_def}
\end{equation}
Equivalently, if $\Omega(\alpha,\beta)=2\operatorname{Im}\langle\beta,\alpha\rangle$ denotes the canonical symplectic form on the one-particle soft phase space, then
\begin{equation}
       \Theta^{H,(0)}_{123}
       =
       {1\over2}
       \bigl[
       \Omega(\alpha_1^{H,(0)},\alpha_2^{H,(0)})
       +\Omega(\alpha_2^{H,(0)},\alpha_3^{H,(0)})
       +\Omega(\alpha_3^{H,(0)},\alpha_1^{H,(0)})
       \bigr] .
       \label{eq:symplectic_area_holonomy}
\end{equation}
Equation~\eqref{eq:symplectic_area_holonomy} denotes the oriented symplectic area of the triangle with vertices $\alpha_1^{H,(0)},\alpha_2^{H,(0)},\alpha_3^{H,(0)}$ in horizon soft phase space.  Since $\alpha_a^{H,(0)}=P_H\mathcal K_\Omega J_a$, this phase is a functional of the three branch currents and the horizon projection.  No branch-local phase assignment reproduces it.

\begin{proposition}[Holonomy tests of the leading horizon channel]
In the coherent leading-soft factorized regime the normalized horizon holonomy obeys
\begin{align}
       \mathcal U^{H,(0)}_{132}
       &=
       \bigl(\mathcal U^{H,(0)}_{123}\bigr)^{-1},
       \label{eq:orientation_reversal_holonomy}\\
       \mathcal U^{H,(0)}_{123}[\alpha_a+\beta]
       &=
       \mathcal U^{H,(0)}_{123}[\alpha_a],
       \label{eq:common_mode_holonomy}\\
       \mathcal U^{H,(0)}_{123}\mathcal U^{H,(0)}_{134}
       &=
       \mathcal U^{H,(0)}_{124}\mathcal U^{H,(0)}_{234}
       \label{eq:four_branch_holonomy_identity}
\end{align}
whenever the displayed phases are well-defined.  Moreover, if the horizon record were equivalent to real pairwise attenuation augmented by removable branch phases, such that
\begin{equation}
       G^H_{ab}=r_{ab}e^{\ii(\varphi_a-\varphi_b)},
       \label{eq:coboundary_phase_model}
\end{equation}
then $\Theta^H_{123}=0$ for every triple.  A residual nonzero $\Theta^{H,(0)}_{123}$ following exterior erasure therefore falsifies that scalar dephasing model.
\end{proposition}

\begin{proof}[Argument]
Orientation reversal complex-conjugates the Hermitian Gram product.  Common-mode invariance follows from the cancellation of all terms containing $\beta$ in the symplectic-area expression~\eqref{eq:symplectic_area_holonomy}.  The four-branch identity reflects the fact that the oriented area of a quadrilateral is independent of its triangulation.  Finally, Equation~\eqref{eq:coboundary_phase_model} yields $\arg(G_{12}G_{23}G_{31})=(\varphi_1-\varphi_2)+(\varphi_2-\varphi_3)+(\varphi_3-\varphi_1)=0$.
\end{proof}

Complete positivity supplies an independent, modulus-dependent constraint. Defining $r_{ab}=|G^H_{ab}|$ and $\Theta^H_{123}=\arg(G^H_{12}G^H_{23}G^H_{31})$, the positivity of the $3\times3$ principal Gram minor requires
\begin{equation}
       1
       +2r_{12}r_{23}r_{31}\cos\Theta^H_{123}
       -r_{12}^2-r_{23}^2-r_{31}^2
       \ge0 .
       \label{eq:theta_gram_bound}
\end{equation}
This inequality provides a direct falsifiability criterion for phenomenological models of horizon decoherence. Specifically, three fitted pairwise visibilities and a fitted cyclic phase that violate Eq.~\eqref{eq:theta_gram_bound} fail the completely positive (CP) Schur-channel test. In the $\Iplus/H$ factorized regime,
\begin{equation}
       \mathcal B_{123}
       =
       \mathcal B^{\Iplus}_{123}\mathcal B^{H}_{123},
       \qquad
       \mathcal B^{H}_{123}=G^H_{12}G^H_{23}G^H_{31}.
       \label{eq:Bargmann_factorized}
\end{equation}
Exterior monitoring and feed-forward can alter the exterior factor, but a residual $\mathcal B^H_{123}$ remains as the cyclic invariant of the complementary horizon output.  This yields a sharp three-branch observable of the horizon channel that remains invisible in a purely pairwise, two-path analysis.

\section{Operational representations and channel consequences}

The leading horizon Schur channel admits multiple operational interpretations. These are frequently discussed in isolation as local noise, detector dependence, adiabaticity, finite-time recovery, or material mimicry. Within the present formulation, however, these aspects constitute a unified framework: each statement specifies a particular representation or control mechanism for the identical completely positive map acting on the branch code.

\subsection{Kernel representations}

The horizon and local descriptions yield the identical channel kernel.  Let $f^\mu=J_a^\mu-J_b^\mu$ denote the current difference of two branch histories, which is compactly supported within Alice's controlled experiment.  Let $\Delta=G_{\mathrm{adv}}-G_{\mathrm{ret}}$ denote the causal propagator, and let $\K_\Omega$ denote the one-particle projection associated with the chosen quasifree state $|\Omega\rangle$.  The coherent radiation generated by the retarded solution possesses the one-particle wavefunction
\begin{equation}
       \alpha_{ab}=\K_\Omega G_{\mathrm{ret}}f .
\end{equation}
Upon the source returning to an identical final Coulomb configuration on both branches, the late-time radiative solution can be equivalently expressed via the causal propagator, yielding
\begin{equation}
       \|\alpha_{ab}\|^2
       =
       \|\K_\Omega\Delta f\|^2
       =
       \langle\Omega|A^{\inm}(f)A^{\inm}(f)|\Omega\rangle .
       \label{eq:local_identity}
\end{equation}
This equation constitutes the local two-point-function formula, representing the same positive quadratic form evaluated prior to propagation to the asymptotic radiative surface.

In a black-hole background, the one-particle radiative space decomposes, within the asymptotic regime employed throughout this work, into exterior and horizon components,
\begin{equation}
       \Hil^{(1)}_{\rad}
       \simeq
       \Hil^{(1)}_{\Iplus}\oplus\Hil^{(1)}_{\Hplus} .
\end{equation}
Let $P_{\hor}$ and $P_{\Iplus}$ denote the projectors onto these respective components.  The two contributions to the coherent displacement norm are given by 
\begin{equation}
       \|\alpha_{ab}^{\hor}\|^2
       =
       \|P_{\hor}\K_\Omega\Delta f\|^2,
       \qquad
       \|\alpha_{ab}^{\Iplus}\|^2
       =
       \|P_{\Iplus}\K_\Omega\Delta f\|^2 .
       \label{eq:projected_norms}
\end{equation}
The local kernel incorporates both terms through the state and the geometry, whereas the horizon calculation isolates the component residing in the inaccessible complementary algebra.

\begin{proposition}[Channel equivalence] In the linear-source, Gaussian-field regime, three distinct quantities determine the identical branch-pair decoherence exponent: the norm of the difference between the horizon coherent states generated by the retarded field, the horizon-projected one-particle norm $\|P_{\hor}\K_\Omega\Delta(J_a-J_b)\|^2$, and the corresponding positive part of the local two-point function smeared with $J_a-J_b$.  The horizon-radiation and local-noise descriptions thus constitute two representations of the same Schur-channel Gram entry.
\end{proposition}

\begin{proof}[Sketch]
The source shifts the Gaussian field by a classical solution.  The overlap between two displaced Gaussian states depends exclusively on the one-particle norm of the difference of the classical shifts.  The one-particle norm equals the two-point function of the in-field smeared with the source difference.  Projection onto the horizon component yields the complementary-channel contribution.  Substituting $A_\mu$ with the linearized metric perturbation and $J_\mu$ with the stress-energy tensor yields the gravitational analog upon imposing the linearized constraints.
\end{proof}

This proposition clarifies the relationship to the local reformulation of horizon decoherence~\cite{DSWLocal2025}. The local calculation eliminates the horizon as an explicit computational boundary; it does not, however, resolve the underlying algebraic question. The local kernel quantifies the total noise generated by the state and the geometry, whereas the open quantum system (OQS) bipartition specifies which portion of that noise can be conditioned upon or recovered via exterior operations.

\subsection{Exterior instruments and eraser bounds}
A reduced density matrix represents the unconditioned output of an experiment, whereas a detector defines a quantum instrument. Let  $\{M_\xi\}$ denote a positive operator-valued measure (POVM) or measurement model acting on the accessible exterior radiation algebra, yielding outcome $\xi$.  The corresponding instrument on the charged branch space is given by 
\begin{equation}
       \I_\xi(\rho)
       =
       \Tr_{\Hplus}\Tr_{\Iplus_{\unres}}
       \left[
          M_\xi
          U_{\dress}(\rho\otimes|\Omega\rangle\langle\Omega|)U_{\dress}^\dagger
          M_\xi^\dagger
       \right] .
       \label{eq:instrument}
\end{equation}
This instrument satisfies $\sum_\xi \I_\xi=\E_{\out}$.  A detector that resolves branch information at $\Iplus$ alters the unconditioned charged state and provides outcome-dependent feed-forward data. Conversely, a detector that erases exterior branch-identifying information preserves a greater degree of conditioned coherence. In both scenarios, the horizon factor persists as the complementary obstruction.

In a factorized coherent regime, the scalar branch-pair decomposition assumes the form
\begin{align}
       G_{ab}^{\tot}
       & = G_{ab}^{\det(\xi)}G_{ab}^{\Iplus,\unres}G_{ab}^{\hor},\label{eq:three_factor}\\
       G_{ab}^{\hor}
       & = \text{complementary horizon record},\nonumber\\
       G_{ab}^{\Iplus,\unres}
       & = \text{unresolved exterior record},\nonumber\\
       G_{ab}^{\det(\xi)}
       & = \text{conditioned detector record}.\nonumber
\end{align}
The first factor is schematic, as a sharp exterior measurement need not act as a scalar multiplier for every branch pair. This decomposition encapsulates the operational distinction: the exterior factors depend on the detector design and feed-forward protocol, whereas the horizon factor is intrinsic to the complementary output.

The central object remains the Schur channel
\begin{equation}
       (\E_{\hor}\rho)_{ab}=G^{\hor}_{ab}\rho_{ab},
       \qquad
       G^{\hor}_{ab}
       =
       \langle\Phi_b^{\hor}|\Phi_a^{\hor}\rangle .
       \label{eq:payload_schur}
\end{equation}
This expression enforces complete positivity on every finite branch code, constrains exterior erasure, and yields multi-branch Bargmann products.

\begin{proposition}[Exterior quantum-eraser bound]
Let $\I_\xi$ be any instrument acting exclusively on the exterior radiation algebra, incorporating classical feed-forward and any recovery map to Alice's charged system.  For two branch currents $a,b$, the maximum recoverable coherence following exterior conditioning is bounded by the fidelity of the complementary horizon outputs~\cite{Petz,AQEC}:
\begin{equation}
       \frac{|(\mathcal R_\xi\circ\I_\xi(\rho))_{ab}|}{|\rho_{ab}|}
       \leq
       F(\rho_a^{\hor},\rho_b^{\hor})^{1/2} .
       \label{eq:eraser_bound}
\end{equation}
In the coherent pure-state horizon regime, this reduces to
\begin{equation}
       V_{\rm rec}^{\rm ext}(a,b)\leq
       |\langle\Phi_b^{\hor}|\Phi_a^{\hor}\rangle|
       =\exp[-\N^{\hor}_{ab}/2] .
       \label{eq:coherent_eraser_bound}
\end{equation}
\end{proposition}

\begin{proof}[Argument]
Exterior measurements refine the $\Iplus$ output and condition phase corrections on the measured record; they do not, however, act on the complementary horizon density operators $\rho_a^{\hor}$ and $\rho_b^{\hor}$.  The monotonicity of fidelity under quantum channels, combined with the information-disturbance tradeoff for complementary channels, bounds the recovered coherence by the indistinguishability of the inaccessible output.  For pure coherent horizon states, the fidelity equals the squared overlap, yielding Eq.~\eqref{eq:coherent_eraser_bound}.
\end{proof}

The bound established in Eq.~\eqref{eq:coherent_eraser_bound} is the physical realization of the fundamental limits of approximate quantum error correction (AQEC)~\cite{AQEC}. In the language of quantum information, the exterior radiation acts as the syndrome ancilla, the feed-forward protocol constitutes the recovery map, and the complementary horizon output dictates the fundamental fidelity ceiling, analogous to the Petz recovery bound~\cite{Petz}.

This bound elevates the ``black hole as detector'' metaphor to a rigorous Stinespring dilation statement. An interior observer would interpret the identical unitary evolution as coherent information transfer. Alice, restricted to the exterior algebra, observes decoherence precisely because she discards the complementary output.

\subsection{Finite-time control and soft/hard scaling}

Finite-time minimization of horizon-induced decoherence is formulated as a channel-control problem~\cite{DKSW2025}.  Fix a horizon cut $C$ and assume that the branch-dependent radiation emitted prior to $C$ has already been generated.  The objective is to determine an admissible future control $u(t>C)$ for Alice's source that maximizes the final channel fidelity to the identity map on the branch subspace. Denoting the channel generated by the continuation protocol $\E_{u,C\to\infty}$, one obtains
\begin{equation}
       u_{\opt}
       =
       \operatorname*{arg\ max}_{u\in\mathcal U_{\admissible}}
       F_{\mathrm{ch}}\left(
          R_u\circ\E_{u,C\to\infty}\circ\E_{C},
          \Id
       \right),
       \label{eq:recovery_optimization}
\end{equation}
where $R_u$ represents a permissable exterior recovery operation.  Within the coherent-state sector, this optimization maximizes the final overlap of the horizon coherent state with the chosen reference branch, effectively constituting a purification problem for the radiation already emitted prior to $C$.

For a controlled protocol exhibiting branch coherence
\begin{equation}
       \rho_{ab}(t)=g_{ab}(t)\rho_{ab}(0),
       \qquad
       g_{ab}(t)=\exp[-\Gamma_{ab}(t)+\ii\phi_{ab}(t)],
\end{equation}
the complete positive (CP) divisibility of the dephasing channel on the branch subspace requires
\begin{equation}
       \dot{\Gamma}_{ab}(t)\geq0
       \quad\text{for every branch pair in every interval.}
       \label{eq:cp_div_condition}
\end{equation}
Whenever an admissible recovery protocol decreases the final horizon displacement norm over a given interval, the exterior reduced dynamics becomes non-CP-divisible within that interval. This provides a rigorous OQS interpretation of finite-time recovery: the horizon environment retains a memory of the source history rather than functioning as a memoryless absorber.

This same operational framework provides the soft/hard discriminator. Let 
 $\chi_\Lambda(\omega)$ denote a smooth resolution function defined with respect to the relevant horizon or detector time generator.  The horizon exponent decomposes as
\begin{equation}
       \N_{ab}^{\hor}
       =
       \N_{ab}^{\hor,<\Lambda}
       +
       \N_{ab}^{\hor,>\Lambda},
       \label{eq:soft_hard_split}
\end{equation}
with
\begin{align}
       \N_{ab}^{\hor,<\Lambda}
       & =
       \int_0^\infty \dd\omega\,\chi_\Lambda(\omega)
       \,\mathcal J_{ab}^{\hor}(\omega),\label{eq:soft_part}\\
       \N_{ab}^{\hor,>\Lambda}
       & =
       \int_0^\infty \dd\omega\,[1-\chi_\Lambda(\omega)]
       \,\mathcal J_{ab}^{\hor}(\omega).\label{eq:hard_part}
\end{align}
Here $\mathcal J_{ab}^{\hor}$ represents the horizon spectral density of the branch-current difference.  For smooth  switching on-and-off time $T_{\rm sw}$ and dwell time $T$, the leading-soft channel predicts the hierarchy
\begin{equation}
       \N^{\hor,<\Lambda}_{ab}\sim A_{ab}T,
       \qquad
       \N^{\hor,>\Lambda}_{ab}\sim B_{ab}T_{\rm sw}^{-p}
       \quad(p>0),
       \label{eq:soft_hard_scaling_result}
\end{equation}
within a parameter window where $T^{-1}\ll\Lambda\ll T_{\rm sw}^{-1}$ is parametrically resolved.  Equivalently,
\begin{equation}
       \N_{ab}^{\hor,<\Lambda}\sim\Gamma_{ab}^{\hor}T,
       \qquad
       \N_{ab}^{\hor,>\Lambda}
       =O(T_{\open}^{-p})+O(T_{\close}^{-p})+O(\text{nonadiabatic corrections}).
       \label{eq:adiabatic_hierarchy}
\end{equation}
Hard radiation also crosses the horizon; however, it fails this two-parameter scaling test. The invariant content of this bipartition is the positive quadratic form selected by the accessible algebra, the detector response, and the time generator, rather than the coordinate-dependent designation ``soft''. I express this dependence as
\begin{equation}
       \E_{\out}=\E_{\out}[\A_{\det},\chi,\Omega,J_a].
       \label{eq:channel_labeled}
\end{equation}

\subsection{Record geometry and material mimicry}

The channel representation reveals geometric structure extending beyond pairwise visibility.  For an $n$-branch superposition with horizon Gram matrix $G^{\hor}$, I define
\begin{equation}
       r_{\rm eff}(G^{\hor})
       =\frac{(\Tr G^{\hor})^2}{\Tr[(G^{\hor})^2]} .
       \label{eq:effective_rank}
\end{equation}
This effective rank quantifies the number of mutually distinguishable horizon records generated by the branch family.  In a continuum branch limit, $G^{\hor}(x,x')$ becomes a positive integral kernel on configuration space; its spectral decay characterizes the resolution of the horizon environment with respect to Alice's charge configurations.

The same kernel formalism also yields the material-mimicry criterion.  For two environments $E_1,E_2$ to exhibit equivalent decoherence over a chosen family of branch currents $\mathcal C=\{J_a\}$, the Gaussian linear-source regime requires the equality of their noise quadratic forms on all branch differences:
\begin{equation}
       \int (J_a-J_b)_\mu
       \left[N^{\mu\nu}_{E_1}-N^{\mu\nu}_{E_2}\right]
       (J_a-J_b)_\nu
       =0
       \quad\forall a,b\in\mathcal C .
       \label{eq:spectral_mimicry}
\end{equation}
Consequently, black-hole-like electromagnetic decoherence constitutes a universality class of infrared spectral density. Ordinary dissipative matter can reproduce a selected electromagnetic channel provided its low-frequency dipole fluctuations and absorption kernel match the black-hole kernel evaluated on Alice's branch differences~\cite{DSWLocal2025,BiggsMaldacena}. The gravitational analog is more restrictive, as it necessitates the corresponding quadrupolar response; nevertheless, the underlying principle remains spectral: the channel is determined by the low-frequency noise and response kernels, detector resolution, branch protocol, hard absorption, and dressing convention.

This section constitutes the quantum-channel-theoretic core of the paper. It translates the horizon overlap into a set of operational principles: the local and horizon descriptions yield the identical Gram entry; exterior measurements define quantum instruments rather than a unique decoherence parameter; recoverable coherence is bounded by the complementary horizon fidelity; finite-time recovery diagnoses non-Markovian channel control; soft and hard records are separated by dwell-time and switching-time scaling; and material mimicry is reduced to kernel matching. These statements form a unified reduced-channel framework, and the subsequent section introduces the multi-branch Bargmann holonomy as its most discriminating interferometric test.

\section{Concrete predictions and checks}

The open quantum system (OQS) formulation yields consistency checks that transcend mere restatements of two-branch horizon estimates. The primary consequence is a strict consistency condition: once branch-dependent horizon records are treated as vectors within an environmental Hilbert space, the corresponding collection of decoherence factors is no longer mutually independent. For any three branches $a,b,c$, the corresponding Gram matrix must be positive semidefinite,
\begin{equation}
       \det
       \begin{pmatrix}
       1&G_{ab}&G_{ac}\\
       G_{ba}&1&G_{bc}\\
       G_{ca}&G_{cb}&1
       \end{pmatrix}
       \ge0 .
       \label{eq:gram_inequality}
\end{equation}
This inequality is a direct consequence of the channel formalism, yet it remains obscured if one merely considers isolated pairwise visibility losses. It imposes rigorous constraints on phenomenological dephasing models and, in the context of an $n$-branch experiment, generalizes to  the requirement that the full matrix $G$ be positive semidefinite.

The second consequence is inherently phase-sensitive.  In a three-arm charged interferometer, the product
\begin{equation}
       \mathcal B_{123}=G_{12}G_{23}G_{31}
\end{equation}
is invariant under independent rephasings of the branch states. It therefore quantifies the holonomy of the reduced infrared channel over the selected branch triple, rather than measuring an arbitrary, convention-dependent pairwise phase. In the factorized leading-soft regime, the horizon contribution is given by
\begin{equation}
       \mathcal B^{H,(0)}_{123}=G^{H,(0)}_{12}G^{H,(0)}_{23}G^{H,(0)}_{31},
       \qquad
       \mathcal U^{H,(0)}_{123}=e^{\ii\Theta^{H,(0)}_{123}}.
\end{equation}
The phase $\Theta^{H,(0)}_{123}$ corresponds to the symplectic area~\eqref{eq:symplectic_area_holonomy} of the triangle formed by the horizon soft displacements. The operational tests for this phase are independent of the overall normalization of the two-branch visibility: reversing the ordering of the arms maps $\Theta\to-\Theta$, adding a common displacement to all arms leaves $\Theta$ invariant, and a scalar branch-phase model yields $\Theta=0$. Furthermore, the simultaneous modulus and phase data must also satisfy
\begin{equation}
       1
       +2r_{12}r_{23}r_{31}\cos\Theta^{H}_{123}
       -r_{12}^2-r_{23}^2-r_{31}^2
       \ge0,
       \qquad
       r_{ab}=|G^H_{ab}| .
       \label{eq:prediction_theta_bound}
\end{equation}
Consequently, the triangle holonomy yields two falsifiable predictions of the leading horizon channel: a null test against removable branch phases, and a completely positive (CP) determinant test against arbitrary pairwise-fit models.  A nontrivial value of $\mathcal B^H_{123}$ does not, by itself, constitute a soft-algebra associator; rather, it is a Bargmann invariant of the environmental Gram matrix, and thus an observable of the phase and visibility geometry induced by inaccessible horizon records.

This invariant can be measured in a concrete three-arm charged-qutrit experiment.  Alice prepares a charged mesoscopic particle in the path state
\begin{equation}
       |\psi_{\rm in}\rangle
       =
       \frac{1}{\sqrt 3}
       (|1\rangle+|2\rangle+|3\rangle)\otimes |0\rangle_{\gamma},
\end{equation}
where each branch corresponds semiclassically to a controlled current history $J_a^\mu(x)$.  The three wavepackets are adiabatically transported into trapping minima that form a small triangle of size $d$ at a radius $b\gg r_s$, held in this configuration for a duration $T$, and subsequently recombined.  In the clean factorized regime,
\begin{equation}
       \rho^{\rm out}_{ab}
       =
       \frac{1}{3}G^{\Iplus}_{ab}G^{\hor}_{ab}.
\end{equation}
Qutrit tomography performed after recombination reconstructs the off-diagonal entries, and the experimentally accessible cyclic quantity is
\begin{equation}
       \mathcal B_{\rm exp}
       =
       \frac{\rho_{12}\rho_{23}\rho_{31}}{\rho_{11}\rho_{22}\rho_{33}},
       \label{eq:Bexp}
\end{equation}
which reduces to $\mathcal B_{123}$ in the case of equal populations.  If Alice additionally performs idealized exterior soft-photon monitoring and applies the corresponding feed-forward erasure on the accessible $\Iplus$ record, the residual cyclic visibility equals
\begin{equation}
       \mathcal B^{\rm residual}_{\rm exp}=\mathcal B^H_{123}
\end{equation}
in the factorized regime.  Thus the modulus and phase of the residual cyclic interference serve as an operational witness of the horizon complementary horizon channel.

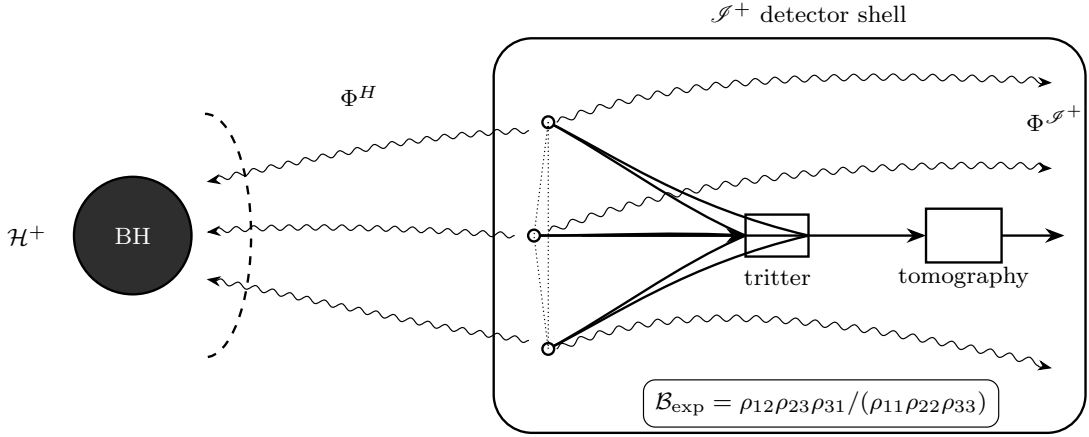
\begin{figure}[t]
\centering
\resizebox{0.96\textwidth}{!}{%
\begin{tikzpicture}[every node/.style={font=\scriptsize}, >=Stealth]
  \fill[black!82] (0,0) circle (0.70);
  \draw[thick] (0,0) circle (0.70);
  \node[text=white] at (0,0) {BH};
  \draw[thick, dashed] (0.86,1.45) arc[start angle=90,end angle=-90,x radius=0.55,y radius=1.45];
  \node[left] at (-0.88,0) {$\Hplus$};

  \draw[thick, rounded corners=13pt] (4.30,-2.35) rectangle (11.35,2.35);
  \node at (8.05,2.66) {$\Iplus$ detector shell};

  \draw[fill=white, thick] (7.30,-0.25) rectangle (8.05,0.25);
  \node at (7.68,-0.5) {tritter};
  \draw[fill=white, thick] (9.45,-0.32) rectangle (10.35,0.32);
  \node at (9.90,-0.5) {tomography};
  \draw[->, thick] (5.35,0) -- (7.30,0);
  \draw[->, thick] (8.05,0) -- (9.45,0);
  \draw[->, thick] (10.35,0) -- (11.10,0);

  \coordinate (A1) at (4.95,1.35);
  \coordinate (A2) at (4.78,0.00);
  \coordinate (A3) at (4.95,-1.35);

  \draw[thick] (7.30,0) .. controls (6.50,0.34) and (5.70,0.95) .. (A1);
  \draw[thick] (7.30,0) .. controls (6.50,0.05) and (5.65,0.00) .. (A2);
  \draw[thick] (7.30,0) .. controls (6.50,-0.34) and (5.70,-0.95) .. (A3);
  \draw[thick] (A1) .. controls (5.75,0.95) and (6.60,0.34) .. (8.05,0);
  \draw[thick] (A2) .. controls (5.65,0.00) and (6.50,0.00) .. (8.05,0);
  \draw[thick] (A3) .. controls (5.75,-0.95) and (6.60,-0.34) .. (8.05,0);

  \filldraw[fill=white, thick] (A1) circle (0.07);
  \filldraw[fill=white, thick] (A2) circle (0.07);
  \filldraw[fill=white, thick] (A3) circle (0.07);
  \draw[densely dotted] (A1)--(A2)--(A3)--cycle;

  \draw[->, decorate, decoration={snake,amplitude=0.38mm,segment length=2.35mm}] (5.05,1.35) .. controls (6.50,1.90) and (8.45,1.92) .. (10.95,1.82);
  \draw[->, decorate, decoration={snake,amplitude=0.38mm,segment length=2.35mm}] (4.95,0.05) .. controls (6.55,0.65) and (8.55,0.88) .. (10.95,0.80);
  \draw[->, decorate, decoration={snake,amplitude=0.38mm,segment length=2.35mm}] (5.05,-1.35) .. controls (6.60,-0.72) and (8.55,-1.05) .. (10.95,-1.58);
  \node[right] at (10.50,1.4) {$\Phi^{\Iplus}$};

  \draw[->, decorate, decoration={snake,amplitude=0.32mm,segment length=2.1mm}] (4.72,1.25) .. controls (3.25,1.20) and (1.80,0.82) .. (0.88,0.64);
  \draw[->, decorate, decoration={snake,amplitude=0.32mm,segment length=2.1mm}] (4.55,0.00) .. controls (3.18,0.10) and (1.78,0.13) .. (0.88,0.04);
  \draw[->, decorate, decoration={snake,amplitude=0.32mm,segment length=2.1mm}] (4.72,-1.25) .. controls (3.25,-1.05) and (1.80,-0.66) .. (0.88,-0.52);
  \node[above] at (2.7,1.4) {$\Phi^H$};

  \node[draw, rounded corners, fill=white, align=center] at (8.20,-2.00)
      {$\displaystyle \mathcal B_{\rm exp}=
       \rho_{12}\rho_{23}\rho_{31}/(\rho_{11}\rho_{22}\rho_{33})$};
\end{tikzpicture}%
}
\caption{Minimal three-arm charged-qutrit interferometer designed to measure the Bargmann invariant.  A charged particle is split by a tritter into three adiabatically controlled current histories $J_1,J_2,J_3$ localized at a small triangular separation $d$ outside the black hole, and subsequently recombined for qutrit tomography.  The wavy lines propagating to the detector shell represent accessible exterior soft records at $\mathscr I^+$; the wavy lines crossing the dashed horizon represent the complementary horizon records.  The reconstructed cyclic invariant is $\mathcal B_{\rm exp}=\rho_{12}\rho_{23}\rho_{31}/(\rho_{11}\rho_{22}\rho_{33})$.  Under ideal exterior monitoring and feed-forward erasure, the residual invariant is $\mathcal B_{\rm exp}^{\rm residual}=\mathcal B^{H,(0)}_{123}$, or equivalently $\mathcal B^H_{123}$ with the leading superscript suppressed, within the factorized regime.}
\label{fig:qutrit_bargmann_setup}
\end{figure}

This channel formulation also rigorously distinguishes between conditioned and unconditioned visibility. Exterior radiation records may be ignored, measured, postselected, or utilized for feed-forward correction, and these distinct choices need not yield identical interference patterns. In the coherent Gaussian regime, the difference is computed from the conditional covariance of the unresolved exterior radiation. The component that rigorously constitutes horizon-induced decoherence is that which remains invariant under alterations to the exterior measurement basis and the recovery procedure.

The formulation also provides a covariance check regarding the designation ``soft.'' Rindler and inertial descriptions must yield the identical final reduced channel once the detector algebra is appropriately transformed, notwithstanding the fact that they decompose the same field into distinct frequency sectors. What varies is the spectral representation of the kernel; the channel itself remains the invariant object. This constitutes the rigorous framework for treating the observer dependence exposed by Killing-horizon examples.

Rotating and charged black holes are naturally accommodated within the same formalism~\cite{GrallaWei,LiRN}. Their effect on charged-sector coherence is characterized by the rank, norm, and degeneracies of the horizon Gram kernel. If an electromagnetic Meissner-like effect suppresses the relevant horizon response in an extremal limit, the OQS formulation asserts not merely that a radiation rate vanishes, but rather that $G^{\hor}_{ab}\to1$ for the corresponding branch-current differences. Finally, a material body mimics the electromagnetic black-hole channel precisely when its low-frequency noise kernel, projected onto Alice's branch-current differences, coincides with the black-hole kernel over the relevant bandwidth. The analogue problem is therefore fundamentally a spectral-kernel matching problem, rather than a search for objects that are superficially analogous to a horizon.

\section{Conclusion}

This paper formulates horizon-induced charged-branch decoherence as the leading-soft sector of inclusive QED restricted to an exterior algebra. Long-range gauge fields propagate branch-identifying information into independent radiative records. The horizon provides a complementary output for a portion of these records, while the fixed-history infrared identity   $\F[J,J]=1$  remains exact. The underlying mechanism is the residual unequal-history Feynman--Vernon/Bloch--Nordsieck factor that remains inaccessible to exterior operations.

The construction commences with an initial product state between the dressed charged branch sector and the independent radiative field, while permitting an arbitrary coherent branch density matrix. The leading soft theorem determines the eikonal displacement generated by the branch currents. Projecting this displacement onto the horizon algebra yields
\begin{equation}
       G_{ab}^{H,(0)}
       =
       \langle\Gamma_b^{H,(0)}|\Gamma_a^{H,(0)}\rangle,
       \qquad
       (\E_H^{(0)}\rho)_{ab}=G_{ab}^{H,(0)}\rho_{ab} .
\end{equation}
This identical Gram entry manifests equivalently as a horizon coherent-state overlap, a horizon-projected one-particle norm, and a local low-frequency noise kernel smeared with the branch-current difference. This equivalence identifies the unified channel underlying both the original horizon-overlap framework and the local-noise formulation.

The black hole constitutes the complementary output of the leading infrared QED channel. Relative to the flat-space soft-QED channel, the key distinction is the causal bipartition of the infrared bath into exterior and horizon records. This framework establishes several operational diagnostics: an exterior quantum-eraser bound, completely positive (CP) divisibility tests for finite-time recovery, multi-branch Gram constraints, soft/hard scaling relations, spectral mimicry criteria, and a charged-qutrit interferometer designed to measure the leading horizon Bargmann invariant.

The triangular holonomy yields the most discriminating prediction of this work. The cyclic product $\mathcal B^H_{123}=G^H_{12}G^H_{23}G^H_{31}$ quantifies a rephasing-invariant symplectic area within the horizon soft phase space. Its properties?orientation reversal, common-mode invariance, the four-branch triangulation identity, and the CP determinant bound~\eqref{eq:theta_gram_bound}?render the leading-soft channel a rigorously falsifiable multi-branch structure. Pairwise visibility fits that violate these identities are invalidated as legitimate horizon Schur channels.

The present analysis remains restricted to the leading soft order. The subleading Low--Burnett--Kroll channel, wherein angular-momentum, spin, and dipole-like records supersede the eikonal current record, constitutes the subject of a separate sequel. Furthermore,  Appendix~\ref{app:asrd} outlines a second avenue for future investigation: exact asymptotic completion recasts the semiclassical horizon record as an asymptotic soft-record decoding problem governed by an encoder, a decoder class, and a target invariant. Both extensions build upon the foundational result established herein: the leading-soft horizon OQS channel and its multi-branch holonomy test.

\section*{Acknowledgement}
Part of this work was carried out while I was participating workshops at the Institute for Pure \& Applied Mathematics (IPAM, USA) and the Simons Institute for the Theory of Computing (SIfTC, USA). This work was supported in part by the U.S. National Science Foundation and the Simons Foundation, by the National Research Foundation of Korea (NRF) (RS-2021-NR060112), and by research funds from Kwangwoon University.

\appendix

\section{Asymptotic soft-record decoding}\label{app:asrd}

This appendix recasts the question of exact-completion as a finite decoding problem~\cite{AQEC}.  The main text establishes a semiclassical statement: the exterior algebra traces out the horizon soft record, thereby discarding the off-diagonal factor $G^H_{ab}$.  An exact theory of quantum gravity can recast this statement as a decoding problem only upon specifying an exact asymptotic encoding map, an input code, an access model, and a decoder class.  This operational organization follows the channel-capacity perspective emphasized by Hayden and Wang in their analysis of the Bekenstein bound~\cite{HaydenWangBekenstein}.  I employ this framework here to isolate the component of the problem that pertains to quantum complexity theory.

Throughout this appendix, the horizon invariants are understood to be leading/eikonal invariants, and the superscript $(0)$ is suppressed.  Let $\mathcal C_n=\{0,1\}^n$ denote a finite branch code.  For each $x\in\mathcal C_n$ let $J_x$ denote an efficiently generated branch current, and let the semiclassical leading-soft horizon channel assign a coherent horizon record $|\Phi_x^H\rangle$.  The target Gram entries and cyclic invariants are
\begin{equation}
       G^H_{xy}=\langle\Phi_y^H|\Phi_x^H\rangle,
       \qquad
       \mathcal B^H_{xyz}=G^H_{xy}G^H_{yz}G^H_{zx}.
       \label{eq:asrd_targets}
\end{equation}
An exact completion supplies a family of asymptotic encoding maps
\begin{equation}
       \Omega_n^{\rm exact}:
       \rho_{\rm br}\longmapsto
       \omega^{\Iplus,{
m exact}}_{\rho,n},
       \label{eq:exact_asymptotic_encoding_map_v6}
\end{equation}
where the output is the exact final state on the full future-null-infinity algebra.  Equation~\eqref{eq:exact_asymptotic_encoding_map_v6} constitutes an additional input.  The semiclassical calculation fixes the invariant in Eq.~\eqref{eq:asrd_targets}; exact completion supplies the data from which a decoder attempts to recover it.

\begin{problem}[Pair ASRD]
The input comprises of branch labels $x,y\in\mathcal C_n$, accuracy parameters $\epsilon,\delta$, a classical description of the current family $\{J_x\}$, and an access model for $\Omega_n^{\rm exact}$.  The task is to compute an estimate $\widehat G_{xy}$ such that
\begin{equation}
       \Pr\bigl[|\widehat G_{xy}-G^H_{xy}|>\epsilon\bigr]\leq\delta .
\end{equation}
\end{problem}

\begin{problem}[Triangle ASRD]
The input comprises of $x,y,z\in\mathcal C_n$ along with the same accuracy parameters and access model.  The task is to compute an estimate $\widehat{\mathcal B}_{xyz}$ such that
\begin{equation}
       \Pr\bigl[|\widehat{\mathcal B}_{xyz}-\mathcal B^H_{xyz}|>\epsilon\bigr]\leq\delta .
\end{equation}
Equivalently, this entails estimating the holonomy phase
\begin{equation}
       \Theta^H_{xyz}=\arg \mathcal B^H_{xyz}.
\end{equation}
\end{problem}

These formulations distinguish between two concepts that are frequently conflated. Information-theoretic completeness asserts that the exact asymptotic state contains a representation of the semiclassical horizon record. The resource-bounded ASRD problem, conversely, asks whether a specified decoder class $\mathfrak D_n$ can extract the invariant using polynomial resources.

\begin{theorem}[$\mathsf{BQP}$ upper bound with coherent circuit access]\label{thm:bqp_upper_asrd}
Assume that $\Omega_n^{\rm exact}$ admits a polynomial-size coherent circuit representation
\begin{equation}
       U_\Omega |x\rangle|0\rangle
       =|x\rangle|\Psi_x^{\Iplus}\rangle
\end{equation}
and that controlled access to $U_\Omega$ and $U_\Omega^\dagger$ is available.  Pair ASRD problem for exact asymptotic overlaps lies in $\mathsf{BQP}$.  The triangle ASRD problem also lies in $\mathsf{BQP}$, requiring only three overlap-estimation calls and polynomial-time classical post-processing.
\end{theorem}

\begin{proof}
Prepare the state
\begin{equation}
       \frac{|0\rangle|x\rangle+|1\rangle|y\rangle}{\sqrt2}|0\rangle
\end{equation}
and apply controlled unitary $U_\Omega$.  A Hadamard test on the control qubit estimates the real and imaginary parts of
\begin{equation}
       \langle\Psi_y^{\Iplus}|\Psi_x^{\Iplus}\rangle
\end{equation}
with additive error $\epsilon$, requiring $O(\epsilon^{-2}\log\delta^{-1})$ repetitions, or achieves this via standard amplitude-estimation refinements when coherent repetition is permitted.  Applying this procedure to  the pairs $(x,y)$, $(y,z)$, and $(z,x)$ yields three complex estimates. Their product yields $\mathcal B_{xyz}$ with polynomial overhead, provided the individual errors are reduced by a constant factor.  The overall computation therefore lies in $\mathsf{BQP}$ under this access model.
\end{proof}

\begin{theorem}[Copy access estimates visibility but not holonomy phase]\label{thm:copy_upper_asrd}
Assume that the access model provides independent copies of $|\Psi_x^{\Iplus}\rangle$ and $|\Psi_y^{\Iplus}\rangle$ but does not permit the coherent preparation of their relative phase.  A swap test estimates
\begin{equation}
       |\langle\Psi_y^{\Iplus}|\Psi_x^{\Iplus}\rangle|^2
\end{equation}
with sample complexity $O(\epsilon^{-2}\log\delta^{-1})$.  The complex phase of $G^H_{xy}$, and hence the triangle holonomy, cannot be determined from copy access alone.
\end{theorem}

\begin{proof}
The swap test accepts with probability
\begin{equation}
       p_{\rm acc}=\frac12\left(1+|\langle\Psi_y^{\Iplus}|\Psi_x^{\Iplus}\rangle|^2\right).
\end{equation}
Standard Chernoff bounds yield the stated sample complexity.  Multiplying either state by an unknown branch-dependent phase leaves the corresponding density matrix $|\Psi_x\rangle\langle\Psi_x|$ invariant, thereby preserving all copy-access statistics, while simultaneously altering the complex overlap phase.  Recovering the Bargmann phase necessitates a coherent reference, an interferometric branch preparation, or an equivalent phase-sensitive access model.
\end{proof}

\begin{problem}[Gram feasibility]
Given a set of approximate complex numbers $\widetilde G_{ab}$ for $a,b\in\{1,\ldots,m\}$ and an accuracy parameter $\eta$, determine whether they are consistent with a physical Schur channel, up to tolerance $\eta$.
\end{problem}

\begin{theorem}[Polynomial-time Gram and holonomy certification]\label{thm:gram_cert_asrd}
The Gram feasibility for a fixed tolerance reduces to a classical polynomial-time positive-semidefinite test.  For $m=3$, defining $r_{ab}=|G_{ab}|$ and $\Theta_{123}=\arg(G_{12}G_{23}G_{31})$, complete positivity requires
\begin{equation}
       1+2r_{12}r_{23}r_{31}\cos\Theta_{123}
       -r_{12}^2-r_{23}^2-r_{31}^2\geq0 .
       \label{eq:asrd_three_branch_gram_bound}
\end{equation}
The orientation and four-branch holonomy identities
\begin{equation}
       \mathcal U_{132}=\mathcal U_{123}^{-1},
       \qquad
       \mathcal U_{123}\mathcal U_{134}=\mathcal U_{124}\mathcal U_{234}
       \label{eq:asrd_holonomy_verifier}
\end{equation}
also constitute polynomial-time checks once the complex Gram entries are provided. 
\end{theorem}

\begin{proof}
The Schur multiplier is completely positive if and only if $G$ is positive semidefinite and $G_{aa}=1$.  Eigenvalue estimation for an explicitly provided $m\times m$ matrix can be performed in polynomial time with respect to $m$ and the input precision.  For $m=3$, expanding the determinant of the Hermitian Gram matrix yields Eq.~\eqref{eq:asrd_three_branch_gram_bound}.  The holonomy identities in Eq.~\eqref{eq:asrd_holonomy_verifier} require only complex multiplication, conjugation, and comparison within a specified tolerance.  They thus define efficient classical verifiers for the finite branch-channel geometry.
\end{proof}

\begin{theorem}[No intrinsic hardness without an access restriction]\label{thm:no_hardness_asrd}
The ASRD problem admits no access-model-independent hardness statement.  Under coherent polynomial circuit access, it lies in $\mathsf{BQP}$ by virtue of Theorem~\ref{thm:bqp_upper_asrd}.  Consequently, any lower-bound claim necessitates a restricted decoder class or a computationally hard exact-completion map.
\end{theorem}

\begin{proof}
A lower bound applicable to all access models would necessarily hold for coherent polynomial circuit access.  However, Theorem~\ref{thm:bqp_upper_asrd} provides a $\mathsf{BQP}$ algorithm under that specific access model.  Therefore,  a valid hardness claim must restrict the decoder resources, the measurement model, the representation of $\Omega_n^{\rm exact}$, or the promise family of black-hole encodings.
\end{proof}

\begin{theorem}[Conditional Harlow--Hayden~\cite{HarlowHayden} reduction]\label{thm:hh_conditional_asrd}
Let $\mathfrak D_n$ be a polynomial-time decoder class.  Suppose that a family of exact asymptotic encodings $\Omega_n^{\rm exact}$ contains a decoding subproblem $\mathsf{HH}_n$ in the following sense: from any instance of $\mathsf{HH}_n$,  one can construct, in polynomial time, branch labels $x,y,z$ and an exact output of $\Omega_n^{\rm exact}$ such that estimating $G^H_{xy}$ or $\mathcal B^H_{xyz}$ to inverse-polynomial precision solves the instance.  Then the ASRD problem is at least as hard as $\mathsf{HH}_n$ under polynomial-time reductions within the decoder model $\mathfrak D_n$.
\end{theorem}

\begin{proof}
By composing the polynomial-time map from the $\mathsf{HH}_n$ instance to the ASRD instance with the assumed ASRD decoder, the resulting output estimate solves the original $\mathsf{HH}_n$ instance by construction.  This constitutes a standard many-one or oracle reduction, depending on whether the decoder is invoked once or adaptively.
\end{proof}

Theorem~\ref{thm:hh_conditional_asrd} formalizes the precise role of the Harlow--Hayden obstruction~\cite{HarlowHayden}.  It does not assert that the semiclassical kernel calculation is computationally hard.  When the branch currents and horizon kernel are explicitly known, one can directly evaluate
\begin{equation}
       \Gamma^H_{xy}=(J_x-J_y)N_H(J_x-J_y),
       \qquad
       G^H_{xy}=e^{-\Gamma^H_{xy}+i\Phi^H_{xy}} .
\end{equation}
Computational complexity arises only after exact completion replaces the missing horizon subsystem with a scrambled asymptotic representation and restricts the decoder. Exponential Gram tomography in isolation yields no cryptographic statement; it merely reflects the dimension of the output table, rather than establishing a computational lower bound. The relevant quantum computational interface comprises access models, $\mathsf{BQP}$ upper bounds, polynomial Gram/holonomy verifiers, and conditional reductions from the chosen exact-completion decoder.

\section{Consistency checks and boundaries of leading-soft channel}
\label{app:scope_checks}

The quantum channel formulation extends beyond a mere two-branch overlap. The pure coherent two-branch overlap constitutes a single matrix element of a completely positive Schur channel. The complete channel formalism imposes Gram positivity, structures detector instruments, separates resolved from unresolved exterior records, and constrains recovery maps via the complementary horizon output. These structural features possess no analogue in an isolated visibility parameter.

The present analysis characterizes softness as dependent on the chosen algebra and flow. A mode is classified as ``soft'' only subsequent to the specification of the relevant time generator and detector response. The invariant object is the reduced channel acting on the branch code, coupled with the specified exterior algebra. The soft/hard decomposition thus constitutes a spectral representation of this channel, rather than an invariant classification of individual quanta.

Hard absorption is incorporated within the same open-system framework and exhibits distinct scaling behavior relative to the leading soft horizon record. This work isolates protocols wherein finite-frequency radiation is driven by nonadiabatic switching, whereas the low-frequency horizon component scales with the dwell time. The two-parameter scaling test rigorously separates these contributions.

The local two-point formulation eliminates the horizon as an explicit computational boundary while preserving its role as a complementary algebra. Both the local kernel and the horizon-overlap calculation evaluate the identical positive quadratic form. The open quantum system (OQS) formulation then addresses which component of this form corresponds to records inaccessible to exterior control.

Finite-time decoherence is rendered unambiguous upon the specification of the operational task. A finite horizon cut, a recovery map, and a detector instrument collectively define a quantum channel. Consequently, the final reduced density matrix and the maximum recoverable visibility become well-defined channel quantities.

Ordinary matter can reproduce the electromagnetic effect provided its low-frequency response matches the horizon kernel evaluated on the branch-current family under consideration. The universal characteristic of the black hole is therefore a specific spectral kernel, rather than the categorical assertion that any environment possessing internal degrees of freedom is equivalent to a horizon.

The assumption of an initial product state in the soft-QED derivation factorizes the dressed charged sector from the independent radiation sector. The spatial superposition experiment then prepares the specific charged-sector density matrix upon which the channel acts. The derivation of the channel and the interferometric probing of its off-diagonal entries thus constitute compatible stages of a unified construction

The cyclic product $\mathcal B_{123}$ constitutes a Bargmann invariant of the environmental Gram matrix.  It quantifies rephasing-invariant channel holonomy and the projective phase geometry.  A nonzero cyclic phase represents a genuine multi-branch observable, as the individual branch phases form a coboundary and cancel over the closed triangle.  Questions of associativity necessitates an actual associator or three-cocycle for the relevant soft transformations; the present calculation establishes the Bargmann holonomy and its complete positivity (CP) constraints.

Exact asymptotic completeness pertains the exact quantum-gravitational algebra at $\Iplus$.  The semiclassical channel employed in this work traces over the ordinary horizon radiative output within the Feynman--Vernon calculation.  An exact reconstruction map derived from exterior soft data would define a separate distinct decoding problem; it would not, however, alter the semiclassical complementary-channel statement.

This work adopts the leading-soft OQS logic from flat-space QED, rather than its flat-space kinematics.  It utilizes fixed-history normalization, off-diagonal environmental distinguishability, kernel-controlled decoherence, and the noncommutativity of infrared and long-time limits.  Upon transitioning to a black hole exterior, It discards the finite-cavity spectral model, the flat momentum basis, and the electron-unparticle interpretation.  The horizon problem replaces these elements with branch currents, a detector algebra, and a complementary horizon output.

\section{Coherent-state overlap}

For bosonic coherent states $|\alpha\rangle$ and $|\beta\rangle$ in a one-particle Hilbert space,
\begin{equation}
       \langle\beta|\alpha\rangle
       =
       \exp\left[-\half\|\alpha\|^2-\half\|\beta\|^2+\langle\beta,\alpha\rangle\right]
       =
       \exp\left[-\half\|\alpha-\beta\|^2+
       \ii\operatorname{Im}\langle\beta,\alpha\rangle\right].
\end{equation}
Consequently, the decoherence exponent corresponds to the norm of the difference of the branch displacements.  In QED, this one-particle norm is evaluated exclusively on transverse (radiative) physical data; the constrained Coulomb data are incorporated into the dressed charged branch rather than as an independent bath.

\section{Schur-channel Kraus representation}

Given a positive semidefinite matrix $G$ with a unit diagonal, let $G=V\Lambda V^\dagger$ denote its spectral decomposition.  The Schur multiplier $\rho\mapsto G\circ\rho$ admits the Kraus operators
\begin{equation}
       K_r=\sqrt{\lambda_r}\sum_a V_{ar}|a\rangle\langle a|.
\end{equation}
The action of the channel is then given by
\begin{equation}
       \sum_rK_r\rho K_r^\dagger
       =
       \sum_{a,b}\left(\sum_r\lambda_rV_{ar}V^*_{br}\right)\rho_{ab}|a\rangle\langle b|
       =
       \sum_{a,b}G_{ab}\rho_{ab}|a\rangle\langle b|.
\end{equation}
Trace preservation follows directly from the condition $G_{aa}=1$.

\section{Relation to infrared-finite QED}

The standard Bloch--Nordsieck result is frequently presented~\cite{BlochNordsieck,YFS,WeinbergSoft,KulishFaddeev} as the cancellation between virtual and unresolved real soft factors in transition probabilities. Within the OQS framework, this is equivalent to the statement that a hard charged observable is not governed by an exclusive soft final state. Rather, it is governed by a reduced channel obtained after tracing over the detector-unresolved soft modes. This reduced channel remains finite once the detector resolution is properly incorporated~\cite{Rey2026a,Rey2026b}.

The leading horizon-restricted formulation differs in that the exterior observer's inclusive trace does not extend toe $\Hplus$.  If $B_{ab}^{\Iplus}(\Lambda)$ denotes the unresolved exterior soft exponent and $B_{ab}^{\hor}$ denotes the horizon exponent, the schematic form is 
\begin{equation}
       \rho_{ab}^{\out}
       =
       \rho_{ab}^{\inm}
       \exp[-B_{ab}^{\Iplus}(\Lambda)]
       \exp[-B_{ab}^{\hor}]
       \exp[\ii\Phi_{ab}] .
\end{equation}
Refinement to the exterior detector can reduce or condition on $B_{ab}^{\Iplus}$; however, they cannot eliminate  $B_{ab}^{\hor}$.

\end{document}